\documentclass[12pt,a4paper]{article}

\usepackage[utf8]{inputenc}
\usepackage{amsmath}
\usepackage{amsfonts}
\usepackage{amssymb}
\usepackage{setspace}
\usepackage{amsthm}
\usepackage{rotating}
\usepackage[a4paper, left=2.5cm, right=2.5cm, top=2.8cm, bottom=2.8cm]{geometry}
\usepackage{float}
\usepackage{authblk}
\usepackage{graphicx}

\usepackage[usenames, dvipsnames]{color}
 \usepackage{tabu} 
 \usepackage{booktabs} 
\usepackage{subfigure,bm}
\usepackage{listings} 
\usepackage{color}
\usepackage{hyperref}
\usepackage{fullpage}
\usepackage{times}
\usepackage{fancyhdr,graphicx,amsmath,amssymb}
\usepackage[ruled,vlined]{algorithm2e}
\theoremstyle{remark}
\usepackage{soul}
\usepackage{comment}

\usepackage{tabularx}
\usepackage{marginnote} 

\usepackage[usenames, dvipsnames]{xcolor}
\usepackage{arydshln}
\definecolor{darkblue}{HTML}{2C326E}
\usepackage{marginnote} 
\usepackage{sectionbox} 
\definecolor{sectboxrulecol}{rgb}{0,0,0.5}
\definecolor{sectboxfillcol}{rgb}{0.95,0.95,1}

\definecolor{subsectboxrulecol}{rgb}{0,0.5,0}
\definecolor{subsectboxfillcol}{rgb}{0.95,1,0.95}

\definecolor{subsubsectboxrulecol}{rgb}{0.5,0.5,0}
\definecolor{subsubsectboxfillcol}{rgb}{1,1,0.9}

\shadowsectionbox 

\usepackage[square,numbers]{natbib}
\newtheorem{proposition}{Proposition}

\theoremstyle{remark}

\newtheorem{remark}{Remark}

\lstdefinestyle{mystyle}{
backgroundcolor=\color{backcolour},
commentstyle=\color{codegreen},
keywordstyle=\color{magenta},
basicstyle=\small,
breaklines=true,
keepspaces=true,
numbers=left,
numbersep=5pt,
}
\usepackage{enumitem}
\definecolor{codegreen}{rgb}{0,0.6,0}
\definecolor{backcolour}{rgb}{0.95,0.95,0.92}

\graphicspath{{figures/}} 

\title{Efficient spline orthogonal basis for representation of~density functions}
\author[1]{Jana Burkotov\'{a} \thanks{corresponding author: jana.burkotova@upol.cz}}
\author[1]{Ivana Pavl\r{u}}
\author[2]{Hiba Nassar}
\author[1]{Jitka Machalov\'{a}}
\author[2]{Karel~Hron}
\affil[1]{Department of Mathematical Analysis and Applications of Mathematics, Faculty of Science, Palack\'y University Olomouc, 17. listopadu 12, Olomouc, 77146, Czech Republic}
\affil[2]{Department of Applied Mathematics and Computer Science Cognitive Systems, Technical University of Denmark, Richard Petersens Plads, 321, 210~2800 Kgs. Lyngby, Denmark}
\date{}
\begin{document}
\maketitle

\begin{abstract}
Probability density functions form a specific class of functional data objects with intrinsic properties of scale invariance and relative scale characterized by the unit integral constraint. The Bayes spaces methodology respects their specific nature, and the centred log-ratio transformation enables processing such functional data in the standard Lebesgue space of square-integrable functions. As the data representing densities are frequently observed in their discrete form, the focus has been on their spline representation.  
Therefore, the crucial step in the approximation is to construct a proper spline basis reflecting their specific properties. Since the centred log-ratio transformation forms a~subspace of functions with a zero integral constraint, the standard $B$-spline basis is no longer suitable. Recently, a new spline basis incorporating this zero integral property, called $Z\!B$-splines, was developed. However, this basis does not possess the orthogonal property which is beneficial from computational and application point of view.
As a result of this paper, we describe an efficient method for constructing an orthogonal $Z\!B$-splines basis, called $Z\!B$-splinets. The advantages of the $Z\!B$-splinet approach are foremost a computational efficiency and locality of basis supports that is desirable for data interpretability, e.g. in the context of functional principal component analysis. The proposed approach is demonstrated on an empirical demographic dataset.

\end{abstract}

\section{Introduction}

Functional data analysis (FDA) is a statistical framework for processing and analysing data that can be represented as continuous functions. However, they are usually not recorded directly in their functional form but as discrete observations. Therefore, the first step in FDA is to fit the discrete data with a proper representation, preferably using a basis expansion \citep{kokoszka17}, and then to proceed with statistical methods of FDA. 
In order to provide a reliable approximation, the basis functions should reflect the specific properties of the underlying functional data objects. Another important point is also computational efficiency: an orthogonal basis having small local support with a moderate number of computations needed for its construction is desired and highly preferable, in particular with large-scale data sets that arise frequently in applications.

In the context of FDA, probability density functions (PDFs) form a specific class of functional data objects. They appear in many applications such as geosciences \citep{simicek21}, demographics \citep{SFPCA,talska21},  economics \citep{talska20} and more. The natural properties of PDFs call for a proper framework for their approximation. The Lebesgue space $L^2$, a standard space for FDA methods, does not respect the primarily relative information carried by PDFs.
Instead, the so-called Bayes space methodology was introduced in \citep{egozcue06,boogaart10,boogaart14} and used as a common framework for statistical analysis for PDFs. Although data processing and methodological developments are possible directly in the Bayes space $\mathcal{B}^2$ of square-integrable log-functions, one often opts for the equivalent representation of PDFs in the $L^2$ space.
The isometric isomorphism between $\mathcal{B}^2$ and a~subspace of $L^2$ of functions with zero integral $L^2_{0}$ represented by the centered log-ratio (clr) transformation \citep{boogaart14} enables such one-to-one mapping that equips the resulting functions with the zero integral property while maintaining the relative information of the original densities.

To represent the discrete observations of PDFs as functional data objects, it is common to express them through basis expansion. Therefore, a selection of a proper basis is a crucial step in FDA. Some of the most popular approaches involve $B$-spline basis \cite{dierckx, deboor78}, Fourier basis for periodic data \cite{bracewell78},  or wavelets \cite{ramsay05}. $B$-spline basis is a usual choice for data approximation since it utilizes advantageous properties of polynomial functions and benefits from the small local support of individual basis functions. In \cite{EDOB, BasnaNP, basna2023splinets}, a new approach to spline basis selection was presented that used a data-driven knot selection to increase the efficiency of the basis in representing the data.

In general, we can claim that a successful approximation strongly relies on the existence of a suitable class of basis functions reflecting the nature of observed data. Specifically, for clr-transformed PDFs, an appropriate basis is required to maintain the zero integral property. Such a class of spline functions called $Z\!B$-splines was first introduced in  \cite{compositional}.  However, in the context of many FDA methods, including the foremost popular functional principal component analysis (FPCA) \citep{kokoszka17} for dimension reduction, orthogonality of such a spline basis is crucial from the methodological point of view while it is needed to keep it interpretable and computationally efficient. 
Functional principal components can be calculated directly with the spline coefficients when orthogonal spline basis is used, see \cite{ramsay05,SFPCA,talska18}.  Moreover, it is desirable that the resulting orthogonal basis preserves most of the beneficial attributes of the original spline basis and does not violate the characteristic local support property. While this is hardly possible with standard orthogonalization algorithms, the proposed \textit{ZB-splinets} procedure is able to deal with this challenge.

As mentioned, the approximation of discrete data is a crucial step of data preprocessing. Since the underlying PDF is assumed to be smooth, the smoothing spline is used for approximation \cite{dierckx}.
Given that the clr transformation establishes a bijective mapping between $\mathcal{B}^2$ and $L^2_0$ spaces, it seems natural to develop smoothing splines in $L_0^2$ while adhering to the zero-integral constraint. Smoothing splines having zero integral were first introduced in \cite{preprocessing}, where necessary and sufficient conditions for the respective $B$-spline coefficients were provided. In \cite{compositional}, the idea of constructing a smoothing spline as a linear combination of basis functions having zero integral is presented. For this purpose, $Z\!B$-splines, which form a basis of the corresponding space and satisfy the zero integral constraint, were defined. 

The problem of orthogonalization of spline bases has been studied in the mathematical literature from both theoretical and application perspectives and can be performed in different ways. In \cite{mason,Redd,alavi}, the orthogonalization of $B$-spline basis is completed through the well-known Gram-Schmidt process. For uniform periodic splines an orthonormal basis has been considered in \cite{kamada} and a quasi-orthogonal basis with nearly diagonal Gram matrix has been proposed in \cite{periodic}. However, the orthogonalization through the Gram-Schmidt process has the disadvantage of losing the local property of $B$-splines since the total support is increasing at every step of the process. Therefore, a modification was proposed in \cite{Redd} that leads to a smaller total support of the resulting orthogonalized spline basis as discussed in  \cite{splinets}. Furthermore, a new symmetric orthogonalization procedure for $B$-splines was introduced in \cite{splinets}, and an associated R package has been developed \cite{SplinetsP} for this purpose. The process produces a net-like structure of splines called splinets that preserve the locality of $B$-spline bases, in fact, the total support is only slightly increased compared to the original $B$-splines.  

Although the problem of orthogonalization of $B$-splines is well-studied, only initial steps in this direction were done in context of novel $Z\!B$-splines for approximation of PDFs. 
This paper aims to fill this gap. For this purpose, we adapt the splinet approach to $Z\!B$-splines and develop a new efficient orthogonal basis $Z\!B$-splinets.
These splinets form a net of orthonormalized $Z\!B$-splines that preserve both the local support property and the zero integral constraint. 
In this paper, we demonstrate that this method is a more effective orthogonalization approach in contrast to the commonly used Gram-Schmidt method.

The main contribution of the paper is to provide a construction of an efficient orthogonal basis for representing probability density functions. 
Specifically, we 
\begin{itemize}
\setlength{\itemsep}{0pt}%
\setlength{\parskip}{0pt}%
    \item present how the splinet approach is utilized for orthogonalization of $Z\!B$-spline basis;
    \item prove the efficiency of the $Z\!B$-splinets  in terms of size of its total support and number of performed inner products;
    \item demonstrate the advantageous properties of $Z\!B$-splinets on an empirical dataset. %
\end{itemize}

The paper is organized as follows. In Section \ref{Zsplines}, theoretical aspects of approximation of PDFs are summarized. The idea of $Z\!B$-spline approximation in context of the Bayes spaces framework is introduced, as well as the role of smoothing in spline representation of PDFs. Section \ref{orthogonalization} provides the main result of this paper, construction of a new orthogonal basis for PDF representation and its computational efficiency. 
The comparison of the proposed orthogonalization approach with standard methods is done in Section \ref{application} by performing functional principal component analysis with empirical demographic dataset.
Finally, Section \ref{conclusions} concludes and provides the final outlook.

\section{ZB-spline basis}
\label{Zsplines}
In this section, we present the construction of $Z\!B$-spline basis that respects the zero-integral property of clr-transformed PDFs and we utilize this basis for approximation of discretized PDFs with a smoothing spline. 
To provide a comprehensive understanding, we first offer a~brief introduction to the Bayes space framework.

\subsection{Bayes spaces}
Probability density functions can be considered as functional compositions characterized by their scale invariance and relative scale properties. Accordingly, the standard space of Lebesgue measurable functions $L^2$ is not appropriate for densities. Therefore, as mentioned in the Introduction section,  
the Bayes spaces methodology was proposed 
as a unifying framework for analyzing PDFs while maintaining their relative scale properties. 
Accordingly, the operations such as summation, scalar multiplication and inner product 
are formulated in the Bayes space setting so that they reflect the scale invariance of PDFs.%

The Bayes space $\mathcal{B}^2(I)$ is defined as a space of nonnegative functional compositions with a~square-integrable logarithm on a bounded interval $I=[a,b]$ accompanied by the operations of perturbation and powering, which are defined as
$$
(f\oplus g)(x) \, = \, \frac{f(x)g(x)}{\int_If(y)g(y)\,\mathrm{d}y},\quad 
(\alpha\odot f)(x) \, = \, \frac{f(x)^{\alpha}}{\int_If(y)^{\alpha}\,\mathrm{d}y}, \quad \alpha \in \mathbb{R},   \quad x\in I.
$$
Note that rescaling of the resulting PDFs in perturbation and powering is done merely for the purpose of the usual interpretability with unit integral constraint, nevertheless, it is not needed for the operations themselves. 
The Bayes space $\mathcal{B}^2(I)$ with the inner product
$$
\langle f,g\rangle_\mathcal{B}=\frac{1}{2\eta}\int_I\int_I\ln\frac{f(x)}{f(y)}\ln\frac{g(x)}{g(y)}\,\mathrm{d}x\,\mathrm{d}y, \quad \eta=b-a,
$$
is a separable Hilbert space \citep{boogaart14}. Consequently, it possesses a suitable geometric  
structure for statistical analysis of PDFs. 
Furthermore, the construction of the $\mathcal{B}^2$ space allows for an unambiguous representation of the original PDFs in the $L^2$ space. 
This is achieved through the so-called centred log-ratio (clr) transformation, defined for a density $f \in \mathcal{B}^{2}(I)$ as
$$
\mbox{clr}(f)(x)=f_c(x)=\ln f(x)-\frac{1}{\eta}\int_I\ln f(y)\,\mathrm{d}y, \quad x \in I.
$$
The resulting real functions capture relative properties of the original PDFs while they enable their standard statistical processing using tools of FDA in $L^2(I)$. However, the clr transformation implies the additional zero integral condition on 
the resulting clr-transformed density function. Indeed, 
$$
\int_I f_c(x) \mathrm{d}x = \int_I \ln f(x) \mathrm{d}x - \int_I 
 \frac{1}{\eta}\int_I\ln f(y)\,\mathrm{d}y\, \mathrm{d}x= 0.
$$
Therefore, the clr-transformed PDFs are elements of $L_0^2(I)$. In order to provide a reasonable spline representation of clr-transformed PDFs, a proper spline basis having the same essential properties as the underlying PDFs is required. We describe the construction of such a basis in upcoming sections.

\subsection{B-splines}
In this section, we first provide a brief general introduction to $B$-spline representation and their corresponding notation. 
We denote as ${\cal S}_{k}^{\Delta\lambda}[a,b]$
the vector space of splines of degree $k\in\mathbb{N}_0$, defined on a finite interval $I=[a,b]$ with the sequence of knots $\Delta\lambda = \{\lambda_i\}_{i=0}^{g+1}$
given by
$$
\lambda_{0}=a<\lambda_{1}<\ldots<\lambda_{g}<b=\lambda_{g+1}.
$$
The space ${\cal S}_{k}^{\Delta\lambda}[a,b]$ has a finite dimension \cite{dierckx}
$$\dim({\cal S}_{k}^{\Delta\lambda}[a,b]) \;= \; g+k+1,$$ 
where $g$ is the number of inner knots in $\Delta\lambda$ and $k$ is the degree of splines in ${\cal S}_{k}^{\Delta\lambda}[a,b]$.
A~convenient choice of basis functions of ${\cal S}_{k}^{\Delta\lambda}[a,b]$ are $B$-splines \cite{deboor78, dierckx}. $B$-splines of degree $k\in\mathbb{N}$ (order $k+1$)  can be constructed using a recurrent formula
$$
B_{i}^{k+1}(x) \, = \,\frac{x-\lambda_i}{\lambda_{i+k}-\lambda_{i}}B_i^k(x)+\frac{\lambda_{i+k+1}-x}{\lambda_{i+k+1}-\lambda_{i+1}}B_{i+1}^k(x),
$$ 
where the $B$-spline of zero degree is defined as a piecewise constant function
$$
B_{i}^1(x) \, = \, \left\{
\begin{array}{rl}
    1 &\quad\mbox{if}\; x\in[\lambda_{i},\lambda_{i+1})\\
    0 &\quad\mbox{otherwise.}
\end{array}\right.
$$
$B$-splines are advantageous for their nice structural properties. They are nonnegative functions with small local support, and for $k>0$ they are differentiable up to the order $k-1$ on $[a,b]$. 
In order to form a $B$-spline basis of ${\cal S}_{k}^{\Delta\lambda}[a,b]$, additional knots to $\Delta \lambda$ need to be considered.
A~commonly used choice is to consider coincident additional knots
\begin{equation}\label{adknot}
\lambda_{-k}=\cdots=\lambda_{-1}=\lambda_{0}=a, \qquad b=\lambda_{g+1}=\lambda_{g+2}=\cdots=\lambda_{g+k+1}.
\end{equation}
Therefore the corresponding extended sequence of knots  $\Delta \Lambda =\{ \lambda_i\}_{i=-k}^{g+k+1}$ is given by
\begin{equation}\label{exknots}
\lambda_{-k}=\cdots=\lambda_{-1}=\lambda_{0}=a <\lambda_{1}<\ldots<\lambda_{g}<
   b=\lambda_{g+1}=\lambda_{g+2}=\cdots=\lambda_{g+k+1}.
\end{equation}
Then the corresponding $B$-spline basis consists of $B$-splines 
$
B_{i}^{k+1}(x)
$, $i=-k, \ldots, g$ defined on the extended sequence of knots $\Delta \Lambda$ \eqref{exknots} and every spline $s_{k}(x)\in{\cal S}_{k}^{\Delta\lambda}[a,b]$ has a unique representation
\begin{equation*}
    s_{k}\left(x\right)=\sum\limits_{i=-k}^{g}b_{i}B_{i}^{k+1}\left(x\right),
\end{equation*}
which can be written in matrix notation as
$$
s_{k}(x) = \mathbf{B}_{k+1}(x)\,\mathbf{b},
$$
where $\mathbf{B}_{k+1}(x)=\left(B_{-k}^{k+1}(x),\dots,B_{g}^{k+1}(x) \right)$ is 
a vector of $B$-spline functions and $\mathbf{b}$ is the vector of $B$-spline coefficients $\mathbf{b} = \left(b_{-k},\dots,b_{g} \right)^{\top}$.

\subsection{ZB-splines}
For the approximation of clr transformed density functions with the zero integral constraint, the use of the standard $B$-spline basis is not that beneficial because it ignores the condition in $L_0^2(I)$ which causes additional constraints on the spline coefficients as described in \cite{preprocessing}.
To overcome this inconvenience, a different basis for splines having the zero integral property was defined in \cite{compositional}. The proposed basis functions are called $Z\!B$-splines. In what follows, we describe their construction and stress their important properties. 

The $Z\!B$-spline of degree $k\in \mathbb{N}_0$ is defined as a first derivative of $B$-spline
\begin{equation}\label{defZB}
Z_{i}^{k+1}(x):=\frac{\mbox{d}}{\mbox{d}x} B_{i}^{k+2}(x).
\end{equation} 
From this definition of $Z\!B$-splines and the formula for derivation of $B$-splines \cite{dierckx}, the relation%
\begin{equation*}
\label{defZ}
Z_i^{k+1}(x)=(k+1)\left(\frac{B_i^{k+1}(x)}{\lambda_{i+k+1}-\lambda_i}-\frac{B_{i+1}^{k+1}(x)}{\lambda_{i+k+2}-\lambda_{i+1}}\right), \quad k\geq 0
\end{equation*}
follows. Especially for $k=0$, $Z\!B$-spline is a piece-wise constant function
$$
Z_{i}^1(x)=\left\{
\begin{array}{rl}
     \dfrac{1}{\lambda_{i+1}-\lambda_i} & \quad\mbox{if}\; x\in[\lambda_{i},\lambda_{i+1})\\
     & \\
     \dfrac{-1}{\lambda_{i+2}-\lambda_{i+1}} & \quad\mbox{if}\; x\in[\lambda_{i+1},\lambda_{i+2}).
\end{array}\right.
$$
A set of linear $Z\!B$-splines is depicted in Figure \ref{Zsplines1} (right) together with the corresponding antecedent set of quadratic $B$-splines (left). A similar set of quadratic $Z\!B$-splines and cubic $B$-splines is shown in Figure \ref{Zsplines2}.

\begin{figure}[ht]
\centering
\includegraphics[width=0.49\textwidth]{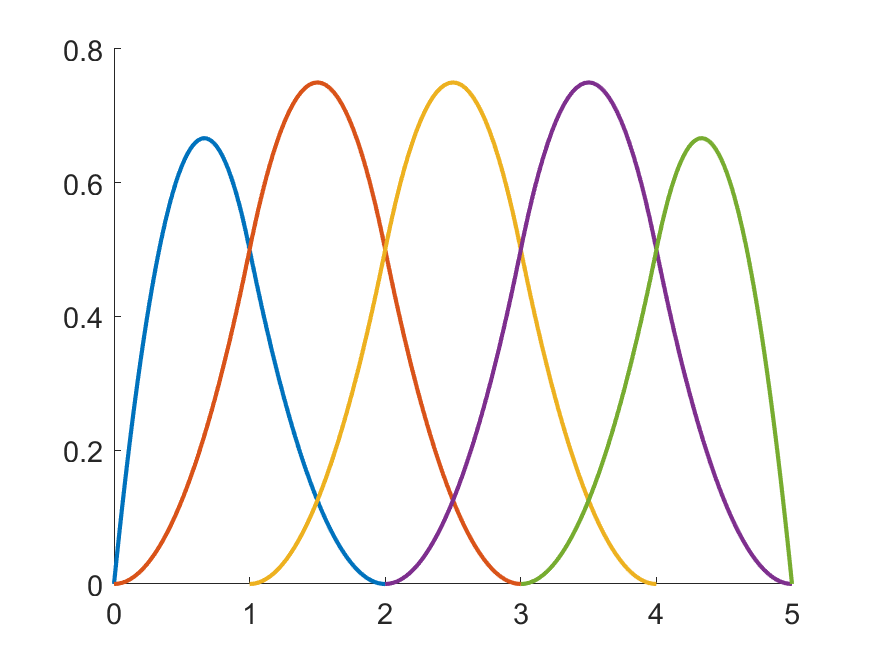}
\includegraphics[width=0.49\textwidth]{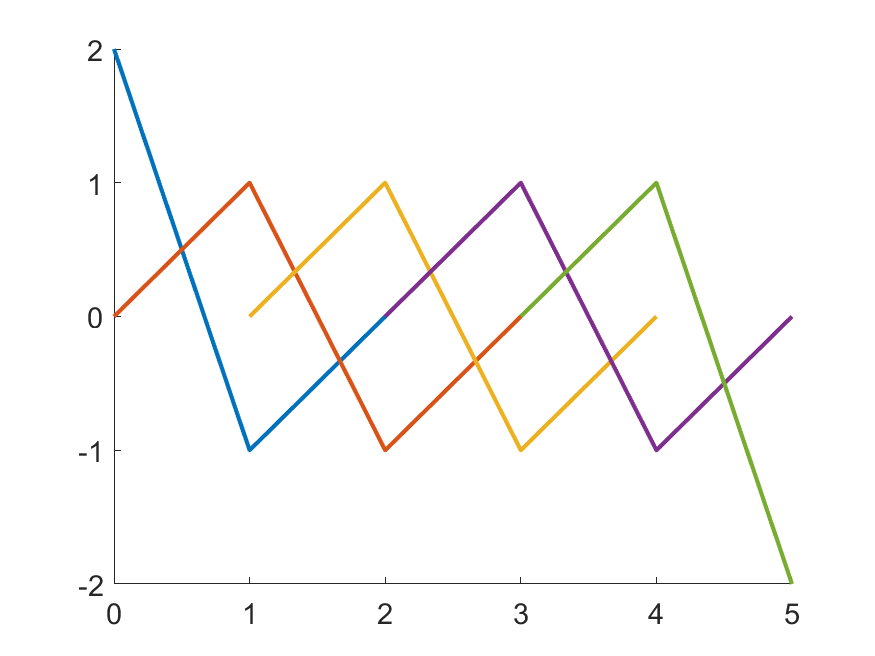}
  \caption{$B$-splines of degree 2 ({\it left}), the corresponding $Z\!B$-splines of degree 1 ({\it right}).}
  \label{Zsplines1}
  \end{figure}

\begin{figure}[ht]
\centering
\includegraphics[width=0.49\textwidth]{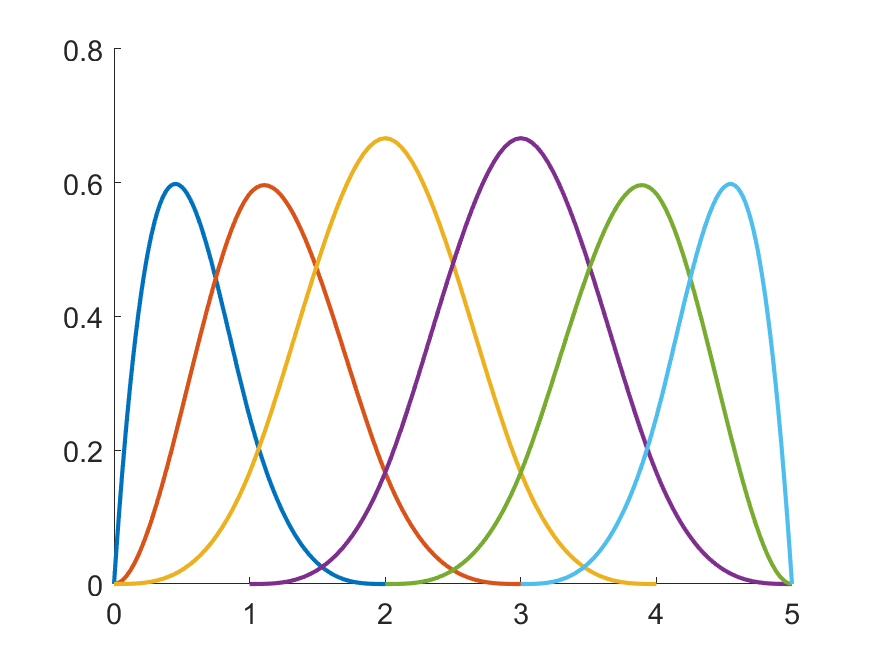}
\includegraphics[width=0.49\textwidth]{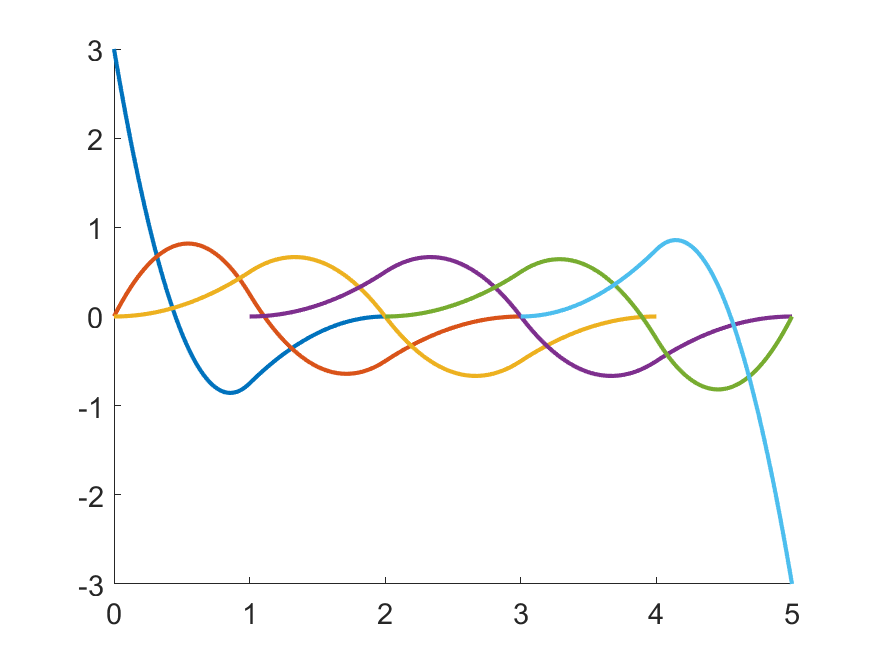}
  \caption{$B$-splines of degree 3 ({\it left}), the corresponding $Z\!B$-splines ({\it right}).}
  \label{Zsplines2}
  \end{figure}
From the construction of $Z\!B$-splines, similar beneficial properties as for $B$-splines are retained for functions $Z_{i}^{k+1}(x)$. $Z\!B$-splines are continuous piecewise polynomial functions of degree $k$ with the local support
$$\mbox{supp}\;Z_i^{k+1}(x) \; = \; \mbox{supp}\;B_i^{k+2}(x)=[\lambda_i,\lambda_{i+k+2}).$$
Moreover, $Z\!B$-splines $Z_i^{k+1}(x)$, $i=-k, \ldots, g\!-\!1$ obey the crucial zero integral property%
$$\int\limits_I
Z_i^{k+1}(x) \, \mbox{d}x \, =\int_a^b  \frac{\mbox{d}}{\mbox{d}x} B_{i}^{k+2}(x) {\mbox{d}x}=B_{i}^{k+2}(b)-B_{i}^{k+2}(a)=0-0=\, 0.$$
This follows from the property of $B$-splines defined on an extended sequence of knots with coincident additional knots \cite{dierckx}
\begin{align*}
B_i^{k+2}(\lambda_i)=0, \quad  -k \leq i \leq 0,\qquad
B_i^{k+2}(\lambda_{i+k+2})=0, \quad  g-k-1 \leq i\leq g\!-\!1. 
\end{align*}
Therefore, $Z\!B$-splines are proper for spline representation of clr transformed density functions. Note that 
coincident additional knots in the knot sequence are crucial because they guarantee that the integral of $Z\!B$-splines over $I=[a,b]$ is zero. Equidistant or periodic additional knots do not guarantee this property. 

In the following,  
${\cal Z}_{k}^{\Delta\lambda}[a,b]$ denotes the vector space of polynomial splines of degree $k\in\mathbb{N}_0$, defined on a finite interval $I=[a,b]$ with the sequence of knots $\Delta\lambda = \{\lambda_i\}_{i=0}^{g+1}$ and having zero integral on $[a,b]$,
\begin{equation*}\label{Zspace}
{\cal Z}_{k}^{\Delta\lambda}[a,b] \, := \, \left\lbrace s_k(x) \in {\cal S}_{k}^{\Delta\lambda}[a,b]\, : \, \int\limits_I s_k(x) \, \mbox{d}x = 0 \right\rbrace \subset {\cal S}_{k}^{\Delta\lambda}[a,b].
\end{equation*}
The dimension of ${\cal Z}_{k}^{\Delta\lambda}[a,b]$ is finite, 
$$\dim({\cal Z}_{k}^{\Delta\lambda}[a,b]) \;= \; g+k,$$
and smaller by one than the dimension of ${\cal S}_{k}^{\Delta\lambda}[a,b]$ due to the zero integral condition, see \cite{compositional}. 
Finally, the set of $g+k$ $Z\!B$-splines $Z_{i}^{k+1}(x)$, $i=-k,\ldots, g-1$ defined on the set $\Delta \Lambda$ \eqref{exknots} with the coincident additional knots 
forms a basis of the space ${\cal Z}_{k}^{\Delta\lambda}[a,b]$, see \cite{compositional} for more details. 
Thus, every spline $s_{k}(x)\in{\cal Z}_{k}^{\Delta\lambda}[a,b]$ has a unique representation
\begin{equation*}
    s_{k}\left(x\right)=\sum\limits_{i=-k}^{g-1}z_{i}Z_{i}^{k+1}\left(x\right),
\end{equation*}
which can be written in a matrix notation as
$$
s_{k}(x) \, = \, \mathbf{Z}_{k+1}(x)\mathbf{z},
$$
where $\mathbf{Z}_{k+1}(x)=\left(Z_{-k}^{k+1}(x),\dots,Z_{g-1}^{k+1}(x)\right)$ is a vector 
of $Z\!B$-spline functions and  $\mathbf{z}$ is a vector of $Z\!B$-spline coefficients $\mathbf{z}= \left(z_{-k},\dots,z_{g-1} \right)^{\top}$.
The relation between $Z\!B$-spline and $B$-spline representations is given by
$$
s_{k}(x) \, = \, \mathbf{Z}_{k+1}(x)\mathbf{z} \, = \, \mathbf{B}_{k+1}(x)\mathbf{b} = \, \mathbf{B}_{k+1}(x)\mathbf{DKz},
$$
where matrices $\mathbf{D} \in\mathbb{R}^{g+k+1, g+k+1}$, $\mathbf{K} \in\mathbb{R}^{g+k+1, g+k}$ follow from \eqref{defZB}:
\begin{align*}
  \mathbf{D}=(k\!+\!1)\,\mbox{diag}\left(\dfrac{1}{\lambda_{1}-\lambda_{-k}},\ldots,\dfrac{1}{\lambda_{g+k+1}-\lambda_{g}}\right)\!, \, 
  \mathbf{K}=\left(\begin{array}{rrrrrr}
	               1 &  0 & 0 & \cdots & 0 & 0 \\
                  -1 &  1 & 0 & \cdots & 0 & 0 \\
                   0 & -1 & 1 & \cdots & 0 & 0 \\
                   \vdots &  \vdots & \ddots & \ddots & \vdots & \vdots \\
                   0 &  0 & 0 &\cdots & -1 & 1 \\
                   0 &  0 & 0 & \cdots & 0 & -1
                 \end{array}\right).
\end{align*}
We refer the reader to  \cite{compositional} for a detailed description.

\subsection{Orthogonalization}

The orthogonalization of the $Z\!B$-spline basis 
offers significant advantages in practical applications.  
We demonstrate a need of an efficient orthogonal process on simplicial functional principal component analysis in Section \ref{application}.

In this section, we provide a general description of orthogonalization. Later in Section \ref{orthogonalization}, we discuss three particular approaches: Gram-Schmidt orthogonalization (one-sided orthogonalization), symmetric Gram-Schmidt orthogonalization (two-sided orthogonalization) and dyadic orthogonalization described in \cite{podgorski} for the construction of an
orthonormal $B$-spline basis with local support. This dyadic orthogonalization approach can also be applied to the construction of the $Z\!B$-spline basis.
Regardless of the specific approach used, orthogonalization can be represented by 
a linear transformation $\Phi$ that maps the $Z\!B$-spline basis to an orthogonal spline basis. 
The new orthogonal basis functions are denoted as $O_{i}^{k+1}(x)$, $i=-k, \ldots, g-1$. This transformation can be converted to matrix formulation 
$$
\mathbf{O}_{k+1}(x) =  \Phi \mathbf{Z}_{k+1}(x),
$$
where $\mathbf{O}_{k+1}(x)=(O_{-k}^{k+1}(x), \ldots,  O_{g-1}^{k+1}(x))$ is a vector of corresponding basis functions that are orthogonal, i.e.
$$
\int_I \mathbf{O}_{k+1}(x)^\top \mathbf{O}_{k+1}(x) \mathrm{d} x = \mathbf{I}.
$$
Then splines $s_k(x) \in  {\cal Z}_{k}^{\Delta\lambda}[a,b]$ can be represented as a linear combination of the new basis $O_{i}^{k+1}(x)$, $i=-k, \ldots, g-1$
\begin{equation}\label{olc}
    s_{k}\left(x\right)= \sum\limits_{i=-k}^{g-1}o_{i}O_{i}^{k+1}\left(x\right) = \mathbf{O}_{k+1}(x) \mathbf{o},
\end{equation}
where $\mathbf{o} = (o_{-k}, \ldots , o_{g-1})^{\top}$ are coefficients of the new spline basis representation.

\subsection{Smoothing splines}\label{smmothing}

The next step in spline approximation of PDFs is to reconstruct the underlying density based on discrete distributional observations, usually resulting from the aggregation of the raw values using histograms \citep{compositional}. The approximation is formulated as a linear combination of basis functions and the quality of the estimate is usually measured by the mean square error. However, since the underlying density is assumed to be smooth, a smoothing spline is employed. The smoothing spline represents a trade-off between achieving the best fit to the observed data and ensuring that the resulting approximation remains smooth.

We suppose to have discrete observations of a probability density function. These data are transformed with a discrete version of the clr transformation and the resulting data points are denoted as $(x_{i},y_{i})$, where $a\leq x_{i}\leq b$, $i=1, \ldots, n$.
The desired smoothing spline $s_{k}\in{\cal Z}_{k}^{\Delta\lambda}[a,b]$ minimizes the following functional,
\begin{equation*}
  \mathcal{J}_l(s_k) \, = \, (1-\alpha)\int\limits_a^b\left[s_k^{(l)}(x)\right]^2\,\mathrm{d} x \, + \, \alpha\sum\limits_{j=1}^{n} w_j\left[y_j-s_k(x_j)\right]^2,
\end{equation*}
where the weights used in the approximation are denoted as $w_{i}>0$, $i=1,\ldots,n$. The parameter $\alpha\in(0, 1]$ is a smoothing parameter that measures the compromise between the best fit to the data and the variability of the approximation controlled by the order of the spline derivative $l\in\left\{1,\ldots,k-1\right\}$. 
The functional $\mathcal{J}_l$ can be rewritten in a matrix notation in the form of a~quadratic function,
\begin{equation*}
  J_l(\mathbf{o}) \, = \, (1-\alpha) \mathbf{o}^{\top} \mathbf{N}_{kl} \mathbf{o} + \alpha [\mathbf{y} - \mathbf{O}_{k+1}(\mathbf{x})\mathbf{o}]^{\top} \mathbf{W}   [\mathbf{y} - \mathbf{O}_{k+1}(\mathbf{x})\mathbf{o}],
\end{equation*}
where $\mathbf{x}=\left(x_{1},\ldots,x_{n}\right)^{\top}$, $\mathbf{y}=\left(y_{1},\ldots,y_{n}\right)^{\top}$,
$\mathbf{w}=\left(w_{1},\ldots,w_{n}\right)^{\top}$, $\mathbf{W}=diag\left(\mathbf{w}\right)$,
$\mathbf{N}_{kl}$ is a~positive semidefinite matrix
$$
\mathbf{N}_{kl}=\left( n_{ij}^{kl}\right)_{i,j=-k}^{g-1}, \quad n_{ij}^{kl}= \int_I (O_i^{k+1}(x))^{(l)} (O_j^{k+1}(x))^{(l)} \mathrm{d} x
$$
and $\mathbf{O}_{k+1}(\mathbf{x})$ is a collocation matrix of orthogonalized $Z B$-splines in given data, i.e.
$$
\mathbf{O}_{k+1}(\mathbf{x}) \, = \, 
\left(
\begin{array}{ccc}
     O_{-k}^{k+1}(x_1) & \cdots &  O_{g-1}^{k+1}(x_1) \\
     \vdots  & & \vdots\\
     O_{-k}^{k+1}(x_n) & \cdots & O_{g-1}^{k+1}(x_n)
\end{array}
\right).
$$
Our objective now is to determine the minimum of the quadratic function $J_l(\mathbf{o})$, which can be written as
\begin{equation}\label{quadr}
  J_l(\mathbf{o}) \, = \, \mathbf{o}^{\top} \left[ (1-\alpha) \mathbf{N}_{kl} \, + \, \alpha\, \mathbf{O}_{k+1}(\mathbf{x})^{\top}  \mathbf{W} \mathbf{O}_{k+1} (\mathbf{x}) \right] \mathbf{o} \, - \, 2\,\alpha \mathbf{o}^{\top}\mathbf{O}_{k+1}(\mathbf{x})^{\top} \mathbf{Wy} \, + \, \alpha\, \mathbf{y}^{\top}\mathbf{Wy}.
\end{equation}
It is obvious that this function has just one minimum if and only if its Hessian is positive definite. From (\ref{quadr}) one can see that
$$
(1-\alpha) \mathbf{N}_{kl} \, + \, \alpha\, \mathbf{O}_{k+1}(\mathbf{x})^{\top}  \mathbf{W} \mathbf{O}_{k+1} (\mathbf{x}) \; \mbox{is p.d.} \quad \Leftrightarrow \quad\mathbf{O}_{k+1}(\mathbf{x}) \; \mbox{is of full column rank.}
$$
The collocation matrix $\mathbf{O}_{k+1}(\mathbf{x})$ is of full column rank iff
there exists $\{u_{-k},\cdots,u_{g-1}\} \subset \{x_1,\cdots,x_n\}$ with $u_i<u_{i+1}$, such that $\lambda_i<u_i<\lambda_{i+k+1}$, $\forall i=-k,\cdots,g-1$.
Then in the case of positive definite function $J_l(\mathbf{o})$ its minimum can be expressed as
\begin{equation}\label{min}
\mathbf{o}^* \, = \, 
\left[ (1-\alpha) \mathbf{N}_{kl} \, + \, \alpha\, \mathbf{O}_{k+1}(\mathbf{x})^{\top}  \mathbf{W} \mathbf{O}_{k+1} (\mathbf{x}) \right]^{-1} 
\mathbf{O}_{k+1}(\mathbf{x})^{\top} \mathbf{Wy}
\end{equation}
and the corresponding smoothing spline $s_k^*(x)$ for given data is 
\begin{equation*}
s_{k}^*\left(x\right) \, = \, \sum\limits_{i=-k}^{g-1}o_{i}^* O_{i}^{k+1}\left(x\right) \, = \, \mathbf{O}_{k+1}(x) \mathbf{o}^*,
\end{equation*}
where $\mathbf{o}^* = (o_{-k}^*, \ldots , o_{g-1}^*)^{\top}$ is the vector of coefficients for related basis functions.

\section{Dyadic orthogonalization}
\label{orthogonalization}

A spline basis can be orthogonalized in several ways. One of the well-known approaches is the Gram-Schmidt (GS) orthogonalization method, which recursively orthogonalizes each element of the basis with respect to an already orthogonal subset of basis functions. However, GS orthogonalization violates the locality of a spline basis with small local supports, since the support of corresponding orthogonalized functions grows up to the entire interval. We refer to the primal GS methods as one-sided GS. The more efficient approach is a modified GS orthogonalization called the symmetric two-sided GS method. The idea is to use one-sided GS applied to the left half of the interval in a right-to-left manner and to the right half of the interval in a left-to-right manner. The remaining basis splines located in the center of the interval are symmetrically orthogonalized with respect to every basis function. For details, we refer to \cite{splinets, mason}. The size of the total support can be further reduced using the splinet approach that is briefly described in the next subsection. Then, we introduce $Z\!B$-splinets as an effective tool for orthogonalizing $Z\!B$-splines. Finally, we discuss the quality of these three approaches in terms of locality and computational efficiency.

\subsection{Splinets}
Splinets as an efficient approach for orthogonalization of the $B$-splines were first introduced in \cite{splinets}. The main advantage is that splinets preserve some of positive properties of $B$-splines, including locality and computational efficiency. The locality is defined by the small size of the spline support and computational efficiency results from the small number of orthogonalization steps required for the $B$-splines to be orthogonal.
Due to its ability to maintain locality, the preferred property of original splines, splinets are prioritized over the other two methods mentioned earlier, namely one-sided GS orthogonalization %
and two-sided GS orthogonalization. 

A better way to visualize this orthogonal basis of splines would be as a net rather than a sequence of orthogonalized functions. Splinets are performed using a recursion algorithm using the support levels that are presented in the form of a dyadic structure. Figure \ref{dyadic1} shows the dyadic structure for the first degree $B$-splines ({\it left}) and the corresponding splinets ({\it right}). The smallest support intervals are at the lower level and increasing support intervals over different layers in the dyadic structure. For $B$-splines of degree $k$, we bind every $k$ adjacent $B$-splines into a group called a $k$-tuplet, then build the dyadic structure. For more details about the splinet algorithm, see \cite{splinets, podgorski}. We note that these references describe the splinet approach applied to $B$-spline with zero boundary conditions.

\begin{figure}[ht]
  \centering

\includegraphics[width=0.45\textwidth,height=0.55\textwidth]{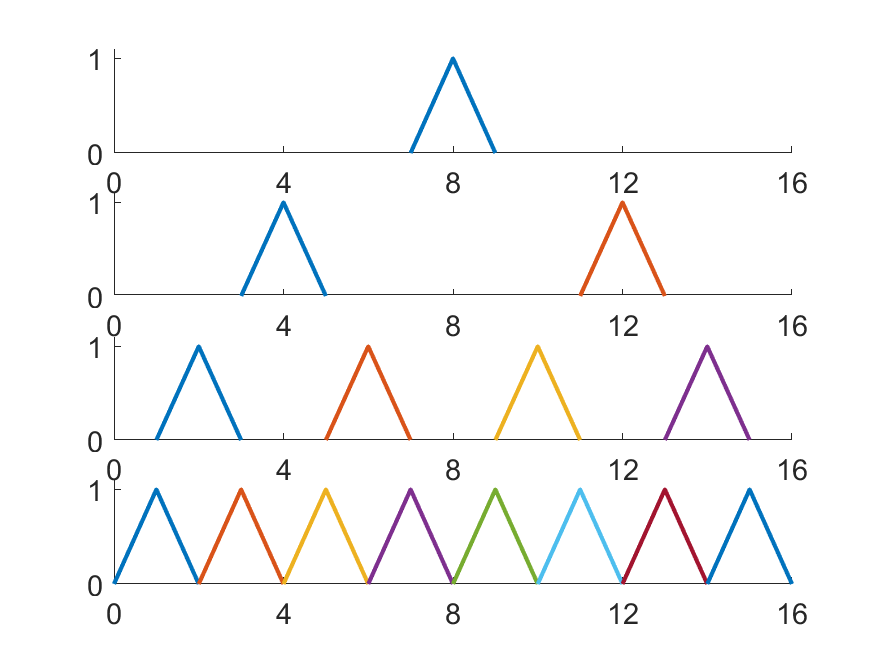}
\includegraphics[width=0.45\textwidth,height=0.55\textwidth]{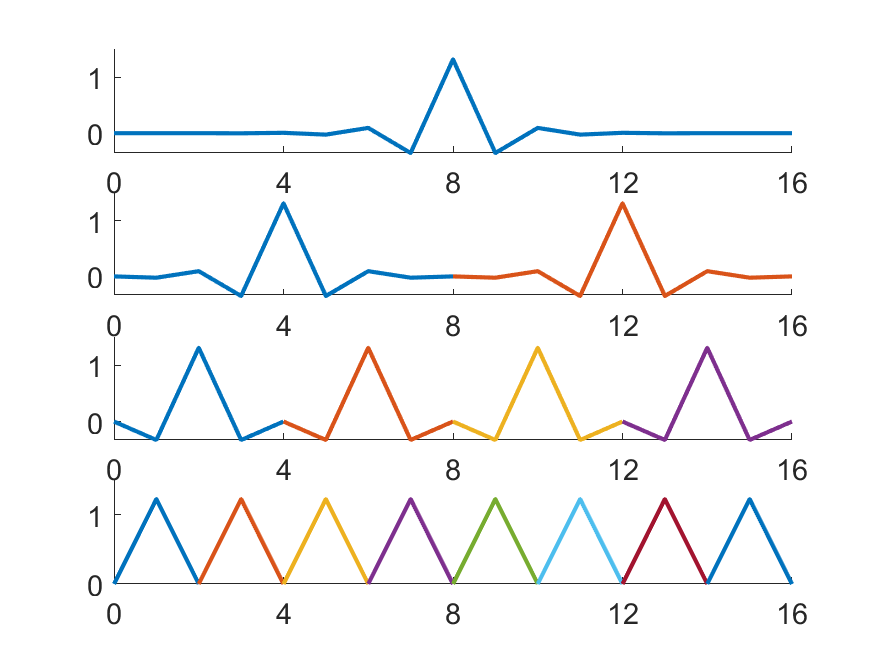}
  \caption{$B$-splines of degree 1 {\it left}, the corresponding $B$-splinets {\it right}}
\label{dyadic1}
\end{figure}

A fascinating result of the splinet algorithm is that one can get a partially orthogonal basis if the algorithm stops at some sufficiently large number $l$ of iterations. In a sense, the partially orthogonal basis is similar to the fully orthogonal basis with a minor error. However, it has the advantage of smaller support and a reduction in the number of inner products required compared to the fully orthogonal basis.

\subsection{ZB-splinets}
In this section, we apply the splinet approach to $Z\!B$-spline basis.
As $Z\!B$-splines of degree $k$ are defined as first derivatives of $B$-splines of degree $k+1$, it is easily realized that for $k \in \mathbb{N}_0$
$$
\operatorname{supp} Z_i^{k+1}(x)=\operatorname{supp} B_i^{k+2}(x)=\left[\lambda_i, \lambda_{i+k+2}\right]. 
$$
To construct the dyadic structure, we determine the number of equispaced inner knots as
$$
g=(2^N-1)(k+1)-k,
$$
where $N > 0$ denotes the number of support levels.
First, we bind every $k+1$ adjacent $Z\!B$-splines of degree $k$ into a $(k+1)$-tuplet (as we bind $k+1$ adjacent $B$-splines of degree $k+1$). Second, we structure the support level $N$ such that the tuplets at each level have disjoint supports. Figure \ref{Zsplintes1} is an example of the dyadic structure of first-degree $Z\!B$-splines, that we bind every two adjacent $Z\!B$-splines into a tuplet, and the support level $N=4$. A similar dyadic structure of quadratic $Z\!B$-splines is presented in Figure \ref{Zsplintes2}.

\subsubsection*{Dyadic orthogonalization}
The algorithm for constructing $Z\!B$-splinets follows these steps:
\begin{enumerate}
    \item[Step 1:] At the lower level, the tuplets have disjoint supports, i.e. each tuplet is orthogonal with respect to the other tuplets at the lower level. It remains to orthogonalize the splines in each tuplet with the symmetric two-sided GS algorithm described below. 
    \item[Step 2:] Orthogonalize all the upper levels with respect to the lower one. With the dyadic structure in place, each tuplet only needs to be orthogonalized with respect to the two adjacent tuplets (one from each side).
    \item[Step 3:] The splines at the lower level are orthogonal to all other splines and form the lower level of the $Z\!B$-splinets structure. By removing the lower level and repeating these steps for the remaining $N-1$ support levels, the entire set of $Z\!B$-splines can be orthogonalized in a recursive manner.
\end{enumerate}

Although an alternative approach can be used to construct this dyadic structure, we found that our choice is a natural way to imitate the splinets algorithm. 

\begin{figure}[ht]
\centering
\includegraphics[width=0.49\textwidth,height=0.55\textwidth]{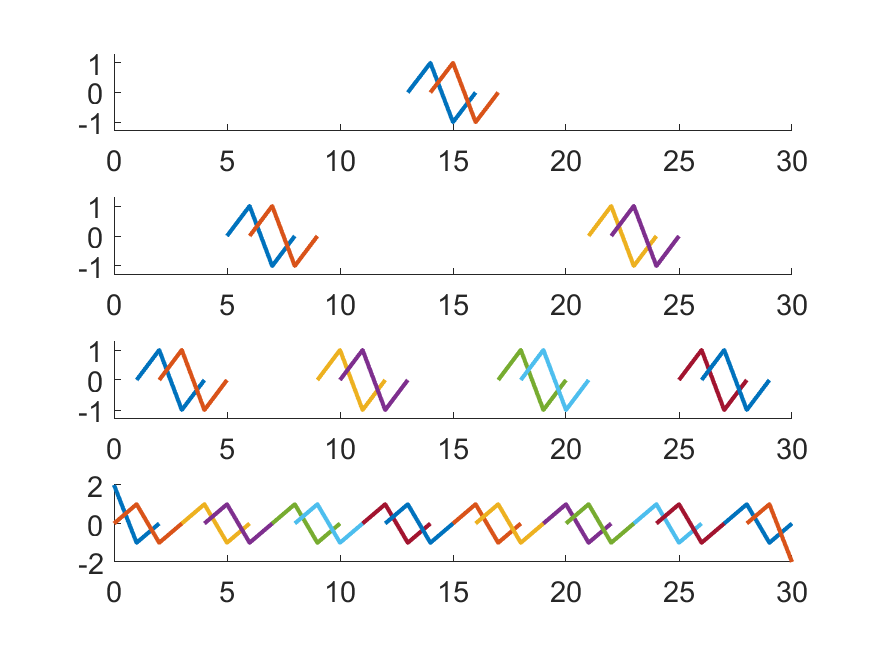}
\includegraphics[width=0.49\textwidth,height=0.55\textwidth]
{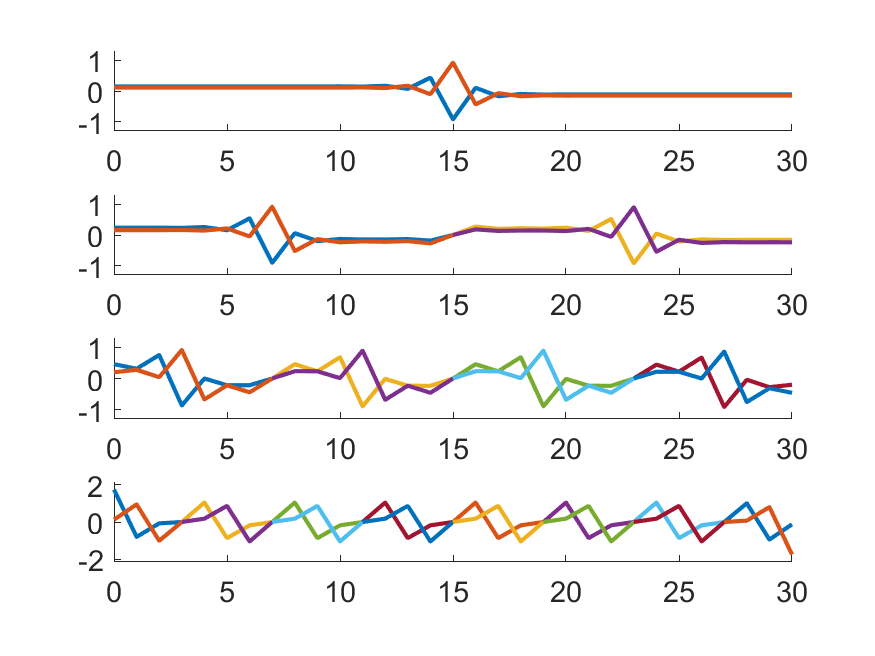}
\caption{$Z\!B$-splines of degree 1 ({\it left}), the corresponding $Z\!B$-splinets ({\it right}).}
\label{Zsplintes1}
\end{figure}

\begin{figure}[ht]
\centering
\includegraphics[width=0.49\textwidth,height=0.55\textwidth]
{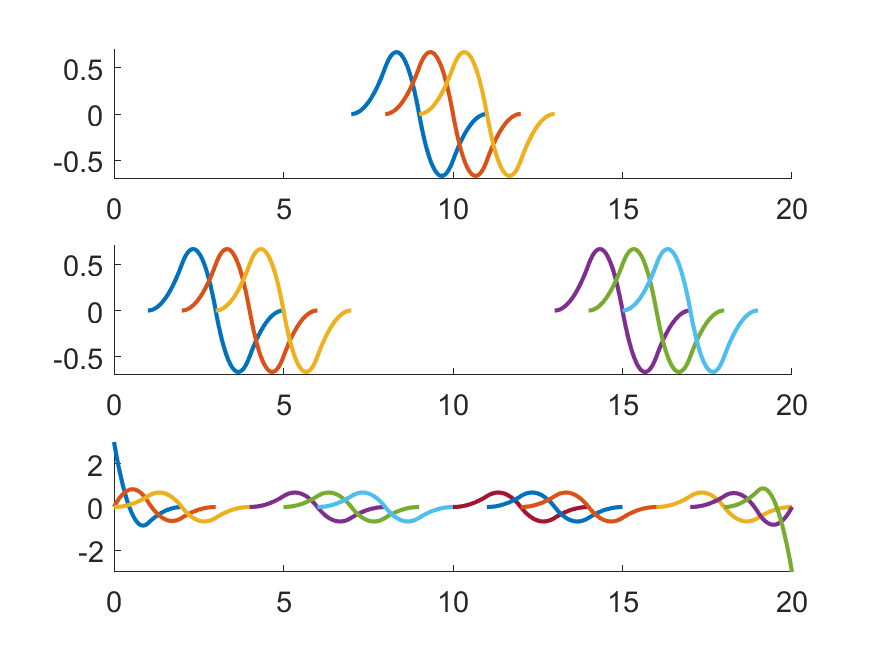}
\includegraphics[width=0.49\textwidth,height=0.55\textwidth]{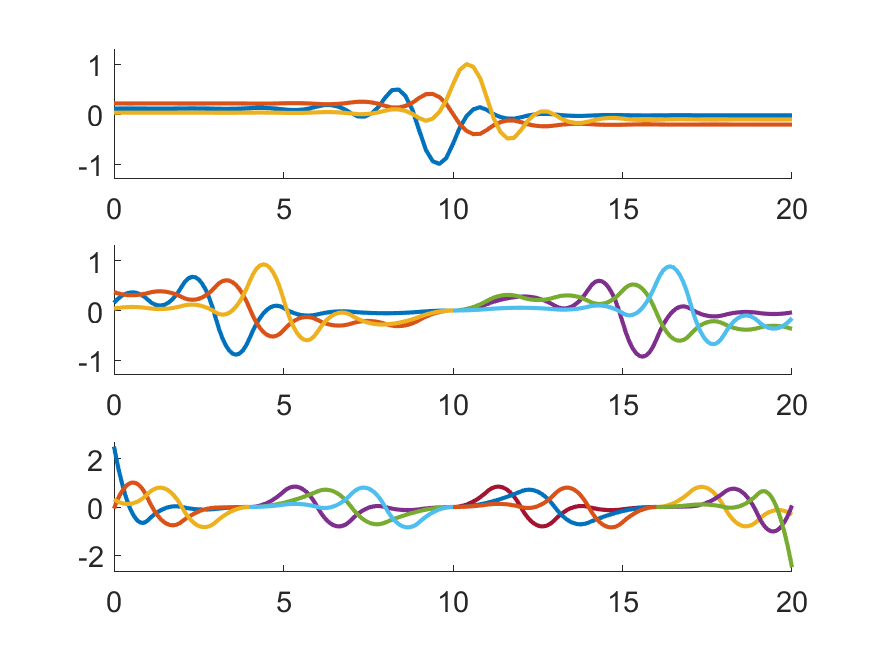}
  \caption{$Z\!B$-splines of degree 2 ({\it left}), the corresponding $Z\!B$-splinets ({\it right}).}
  \label{Zsplintes2}
  \end{figure}

\subsubsection*{Symmetric two-sided Gram-Schmidt}
A modification of the Gram-Schmidt orthogonalization method that preserves the symmetric property was proposed in \cite{Redd}. 
We performed the symmetric two-sided GS to ordered \mbox{$Z\!B$-splines} within the dyadic orthogonalization with respect to a central point which is here chosen as a midpoint of the interval determined by their supports.

\begin{enumerate}
    \item[Step 1:] Perform one-sided GS orthogonalization on $Z\!B$-splines with supports on each of two halves of the interval separately. The left-to-right orthogonalization is performed to the $Z\!B$-splines which have their supports on the left half, and the right-to-left orthogonalization is applied to the $Z\!B$-splines with their supports in the right half of the interval, respectively.
    \item[Step 2:] Orthogonalize remaining central $Z\!B$-splines. Bind the first and the last of the remaining  $Z\!B$-spline and orthogonalize them with respect to already orthogonalized ones. Then, perform the pairwise symmetric orthogonalization as follows
    \begin{align*}
        \tilde z_i &= \left( \frac{1}{\sqrt{1+\langle z_i,z_j \rangle}}  + \frac{1}{\sqrt{1-\langle z_i,z_j \rangle}}\right) \frac{z_i}{2} + \left( \frac{1}{\sqrt{1+\langle z_i,z_j \rangle}}  - \frac{1}{\sqrt{1-\langle z_i,z_j \rangle}}\right) \frac{z_j}{2} \\
        \tilde z_j &= \left( \frac{1}{\sqrt{1+\langle z_i,z_j \rangle}}  - \frac{1}{\sqrt{1-\langle z_i,z_j \rangle}}\right) \frac{z_i}{2} + \left( \frac{1}{\sqrt{1+\langle z_i,z_j \rangle}}  +\frac{1}{\sqrt{1-\langle z_i,z_j \rangle}}\right) \frac{z_j}{2}
    \end{align*}
    where the pair $z_i, z_j$ is the first pair of the remaining central splines, and $\langle z_i,z_j \rangle$ denotes their scalar product. 

Repeat this procedure with a pair of second and last-to-one remaining $Z\!B$-splines and so forth with other pairs of central $Z\!B$-splines until either all central $Z\!B$-splines are orthogonalized (the total number of splines is even) or only one $Z\!B$-spline remains (the total number of splines is odd) and orthogonalize the last one with respect to all other previously orthogonalized splines.
\end{enumerate}

We refer to \cite{splinets} for further details on splinet approach in the context of $B$-splines orthogonalization.

\subsection{Efficiency of the $Z\!B$-splinets}
\label{eff}
The main advantage of $Z\!B$-splinets is that they acquire the locality and computational efficiency of splinets. The locality is indicated by the small size of the total support of $Z\!B$-splinets. For the three different orthogonalizations we have already introduced, we compare the relative total supports, which is the ratio of the total support size of a basis over its domain. Computational efficiency can be assessed by the number of inner products one has to evaluate in the orthogonalization process.
In the sequel, we will limit our discussion to the fully dyadic case and equispaced knots to simplify computations and facilitate a clear discussion. 
The following propositions state these two properties, and the proofs follow those in \cite{splinets} derived for $B$-splines.

\begin{proposition}
Let $\Delta \Lambda=\{\lambda_i\}_{i=-k}^{g+k+1}$ be a dyadic set of knots with equispaced $g$ inner knots \eqref{exknots} where $g= (2^N-1)(k+1)-k$ for $N> 0$. Then the relative total support of orthogonalized $Z\!B$-splines of degree $k\in \mathbb{N}_0$
\begin{itemize}
    \item[(i)]  for one-sided Gram-Schmidt orthogonalization is \ $\displaystyle\frac{g}{2} +k+1-\frac{1}{g+1}$,
    \item [(ii)] for symmetric two-sided Gram-Schmidt orthogonalization is \ $\displaystyle\frac{g}{4} +k+\frac{7}{4} -\frac{2}{g+1}$,
    \item [(iii)] for splinet orthogonalization is \  $\displaystyle  (k+1)\, {\log_2 \left(\frac{g+k}{k+1} +1\right)}$.
\end{itemize}
\end{proposition}
\begin{proof}
    (i) The first $Z\!B$-spline, obtained through orthogonalization using a one-sided Gram-Schmidt approach and defined on $\Delta \Lambda$ with coinciding additional knots,  always has a~relative support equal to $2/(g+1)$. The relative support of the subsequent elements in the basis are $3/(g+1), 4/(g+1), \ldots, 1,1, \ldots, 1$, where the entire interval is a support for the last $k+1$ $Z\!B$-splines orthogonalized with one-sided Gram-Schmidt approach.
    
    (ii) The two-sided orthogonalized $Z\!B$-splines can be obtained by applying one-sided orthogonalization from both sides until the midpoint of the interval. This process yields two sets of splines with an individual relative support sequence of $2/(g+1), 3/(g+1), \dots, 1/2$. Additionally, we need to account for the final $k+1$ splines in the center, which have a relative support of one. Thus, the total support for the two-sided orthogonalized $Z\!B$-splines can be calculated as follows:
    $$
    2\left(\frac{1}{2(g+1)} \left(\frac{g+1}{2}+2 \right) \left( \frac{g-1}{2} \right)   \right)+(k+1)= \frac{g}{4} +k+\frac{7}{4}-\frac{2}{g+1}.
    $$
    
    (iii) The relative total support of the $Z\!B$-splinets can be determined based on the number of levels $N$. At each level, the relative total support of the splinet is always equal to $k+1$, which results in $(k+1)N$. Taking into consideration the relation $N = \log_2 \left(\frac{g+k}{k+1} +1\right)$, the formula for the relative total support follows.
\end{proof}

\begin{remark}
    It is worth noting that two-sided orthogonalization produces a relative total support that is approximately half compared to the one-sided orthogonalization. Both approaches result in a support size on the order of $n$, while the relative total support of the splinet is on the order of $\log(n)$.
\end{remark}

\begin{remark}
    From the previous proposition, it is evident that the total relative support of orthogonalized $Z\!B$-splines, whether obtained through one-sided or two-sided symmetric orthogonalization, depends on the location of the knots. In contrast, with splinets, the total relative support does not depend on the knot placement.
\end{remark}

\begin{proposition}
Consider the dyadic structure case for the $Z\!B$-splines of degree $k$. Then the number of evaluations of inner products needed
\begin{itemize}
\item[(i)] for one-sided orthogonalization is $(k+1)g +k(k-1)/2 -1$,
\item[(ii)] for two-sided orthogonalization is $(k+1)(2g-4) - k(k+1)/2$,
\item[(iii)] for splinet orthogonalization is 
$$
\left(\frac{5k+4}{2}\right)(g+k)-2\, \log_2\left(\frac{g+k}{k+1} +1\right)(k+1)^2 - \frac{k(k+1)}{2}.
$$

\end{itemize}
\end{proposition}

\begin{proof}
  (i) In one-sided orthogonalization, the first $k+1$ splines require $0, 1, \dots, k$ inner products, respectively. Since the supports of $Z\!B$-splines are disjoint, the remaining $g-1$ splines all require $k+1$ inner products.
  Therefore, the total count of inner products can be calculated as $(g-1)(k+1) + 1 + \dots + k = (k+1)g + \frac{k(k-1)}{2} - 1$.
   
   (ii) The number of inner products required for two-sided orthogonalization involves the inner products needed for one-sided orthogonalization of splines from both sides of the interval until the midpoint, followed by the orthogonalization of $k+1$ splines located in the center of the interval.
   In one-sided orthogonalization, the first $k+1$ splines require $k(k+1)/2$ inner products. The remaining $(g+k)-3(k+1) = g-2k-3$ splines need $k+1$ inner products for orthogonalization due to the disjointness of supports discussed in (i). The inner $k+1$ splines are gradually orthogonalized with respect to those already orthogonal ones, i.e. they are orthogonalized to $g-1, g, \ldots, g+k-1$ splines that require $(k+1)(2g+k-2)/2$ inner products.
   Therefore, the total number of inner products needed for two-sided orthogonalization is equal to 
   $$
   2\left( \frac{k(k+1)}{2}\right)+(g-2k-3)(k+1)+ \frac{(k+1)}{2}(2g+k-2) = 
   (k+1)(2g-4) -  \frac{k(k+1)}{2}.
   $$
   
   (iii) For a splinet, orthogonalization of each $k+1$-tuplet requires $k(k+1)/2$ inner products. Thus, the bottom row in the dyadic structure with $N$ rows needs $2^{N-1}k(k+1)/2$ inner products. There are $2^{N-1}-1$ $(k+1)$-tuplets above the bottom row, and each of them needs to be orthogonalized only with respect to two neighboring $(k+1)$-tuplets taken from the top row of the already orthogonalized portion of the dyadic net of splines. Hence, each of these non-orthogonalized 
   $(k+1)$-tuplets requires 
   $2(k+1)^2(2^{N-1}-1)$ inner products. 
   The total number of the inner products that is required at this step of the recurrence is
    $$
    2^{N-1}k(k+1)/2+2(k+1)^2(2^{N-1}-1)=2^{N-2}(k+1)(5k+4)-2(k+1)^2.
    $$
    As a result of each loop, the dimensions of the dyadic structure are reduced by one, see Algorithm 4 in \cite{splinets}. Thus  
\begin{multline*}
\sum_{j=1}^{N-1} \left(2^{N-j-1}\left((k+1)(5k+4)\right)-2(k+1)^2\right)= \frac{(k+1)(5k+4)}{2} \sum_{j=1}^{N-1}2^j-2(N-1)(k+1)^2\\
=\frac{(k+1)(5k+4)}{2}  \left(2^{N}-2\right)-2(N-1)(k+1)^2.
\end{multline*}
The final result follows from the relations $2^N= \frac{g+k}{k+1}+1$ and $N=\log_2(\frac{g+k}{k+1} +1)$. 
\end{proof}

\section{Application}
\label{application}
To showcase the benefits of using an orthogonal basis for representation of clr transformed densities and to compare the different orthogonalizing approaches, an empirical dataset from the field of demography was analysed. First, however, a short introduction to functional principal component analysis is offered to the reader as it is a tool of FDA where the orthogonality of the underlying basis plays a key role.

\subsection{Functional principal component analysis}
In many applications of multivariate and/or high-dimensional statistics, it is necessary to start with reducing the dimensionality of the data, i.e. with expressing the data in a new reduced vector space while keeping a sufficient amount of the original information intact. Principal component analysis (PCA)  handles this by constructing new 
variables as linear combinations of the original ones
which better capture  
the data variability.
These variables, commonly called principal components, are 
obtained as eigenvectors of the sample covariance matrix  and therefore
create a sequence of mutually orthogonal vectors. This further ensures the interpretability of the respective eigenvalues as the amount of variability captured by the corresponding principal component. This way, it is possible to  decrease the number of considered variables (in this case, principal components) while securing the main sources of variability and possibly filtering out potential noise. 

In the context of FDA, one refers to functional principal component analysis (FPCA, \cite{ramsay05}).  For PDFs, a simplicial functional principal component analysis (SFPCA) was introduced in \cite{SFPCA}, where the analysis is performed directly on clr transformed PDFs.
In this case, the newly constructed functional variables (called functional principal components) are derived as the eigenfunctions obtained from the sample covariance operator. Using a spline representation of functional data enables a simplification of this procedure since the matrix of spline coefficients serves as the discrete representation of the data with respect to the given spline basis \cite{ramsay05}. Analogously, the eigenfunctions of the sample covariance operator can be represented equally as the linear combinations of the basis functions. For orthogonal spline basis, the vectors of coefficients correspond to the eigenvectors of the sample covariance matrix built on the data matrix of coefficients.
Thus, using orthogonal basis such as $Z\!B$-splinets directly for the representation of the transformed PDFs reduces the problem to well-known standard multivariate PCA \citep{johnson07} of the coefficient matrix without any additional corrections  
and computations directly using $Z\!B$-splinet coefficients are possible \cite{SFPCA, DPCA}. 

In the multivariate case, a modification of 
PCA, called sparse principal component analysis, enables further simplification of the data structure. This approach alters the eigenvectors of the covariance matrix in such a way that the low absolute values are forced to 0 via a penalized regression using the elastic net penalty \cite{erichson18}. Essentially, the choice of sparsity parameter
indirectly determines the sparsity of the resulting principal components (i.e. the number of zero-valued elements in individual eigenvectors), which helps with the search for the true patterns. This can be seen as the trade-off between the simplicity of the component structure and the amount of explained variability maintained, which naturally decreases with higher values of the sparsity parameter. 
While the functional counterpart of sparse PCA in the classical sense aims to produce functional principal components which are only nonzero in subregions \cite{NIE2020}, we focus here instead on simplifying the linear combination of the basis functions, where $Z\!B$-splinets can profit naturally from their local domain. Considering the data matrix of spline coefficients, this approach can be extended again even for the case of functional data such as PDFs.

\subsection{Demographic example}

In this section, different basis selections are showcased on an empirical demographic example through functional principal component analysis, where an orthogonal basis is beneficial and enables to reduce the FPCA problem to a standard multivariate PCA of a corresponding coefficient matrix.
The aim is to perform a comparison of different orthogonalization approaches applied on the $Z\!B$-spline basis rather than a detailed statistical analysis and interpretation of its results as done in \citep{SFPCA} for the same dataset.
The original dataset describing population age distributions in Upper Austria was obtained as pre-aggregated set in the form of histograms representing the relative frequency of men and women in 19 time intervals with centers $x_i$, $x_i= 2+5(i-1)$, $i=1,\ldots,19$, living in 57 municipalities. For illustration,  two observations corresponding to one selected municipality are depicted in Figure \ref{fig:example_1_JM} (right), where the red and blue symbols represent relative frequencies for women and men, respectively.
The histogram data were transformed with discrete clr transformation and then smoothed by a~quadratic smoothing spline ($k=2$) described in Section~\ref{smmothing}. The smoothing was performed 
with respect to the first derivative of the spline, i.e. $l=1$, the parameter $\alpha$ was fixed to $\alpha = 0.5$ and a sequence of nine equidistant knots on $[a,b]=[0,95]$: $\Delta\lambda \, = \{\ \lambda_i\}_{i=0}^{8}$,
 $\lambda_i = 95i/8$ was used.

The resulting smoothing splines in $L_0^2[0,95]$  for the clr transformation of the two selected instances 
are depicted in Figure \ref{fig:example_1_JM} (left). Using inverse clr transformation (see \cite{boogaart14} for more details) the smoothing splines in the Bayes space are obtained, see Figure~\ref{fig:example_1_JM} (right). The corresponding smoothing splines for the whole dataset consisting of 114 observations are shown in Figure~\ref{fig:example_JM}. 

\begin{figure}[ht!]
\centering
\begin{tabular}{cc}
\includegraphics[width=0.5\textwidth]{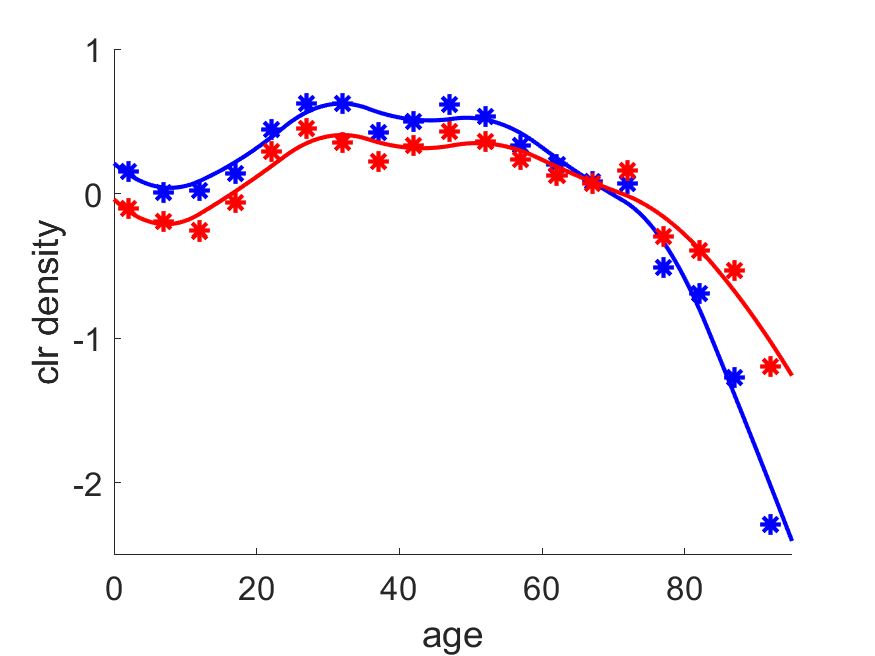}&
\includegraphics[width=0.5\textwidth]{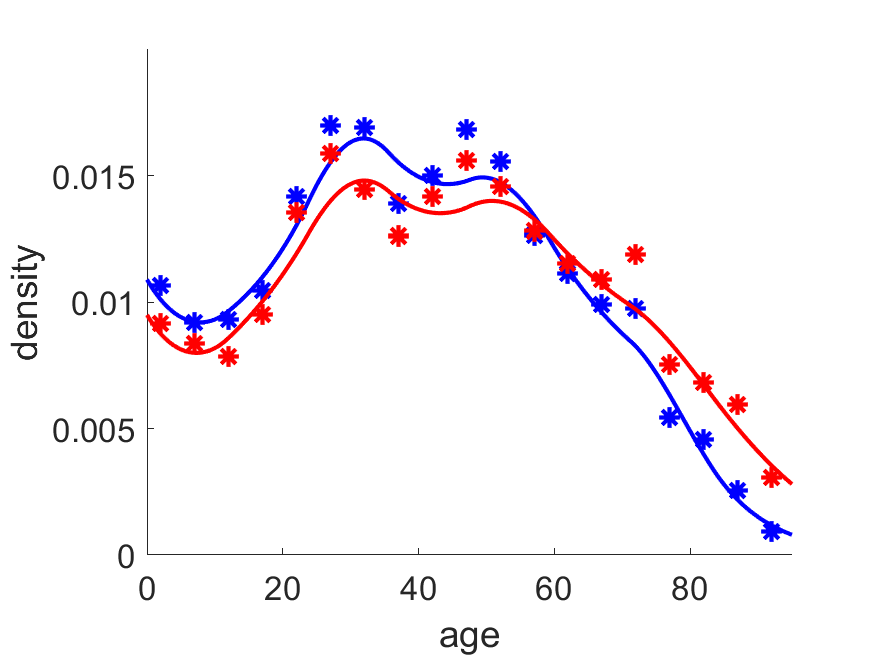}\\
\end{tabular}
\caption{Example of two approximated functional observations  in $L^2_{0}$ (left) and their $\mathcal{B}^{2}$ counterparts (right). 
}
\label{fig:example_1_JM}
\centering
\begin{tabular}{cc}
\includegraphics[width=0.5\textwidth]{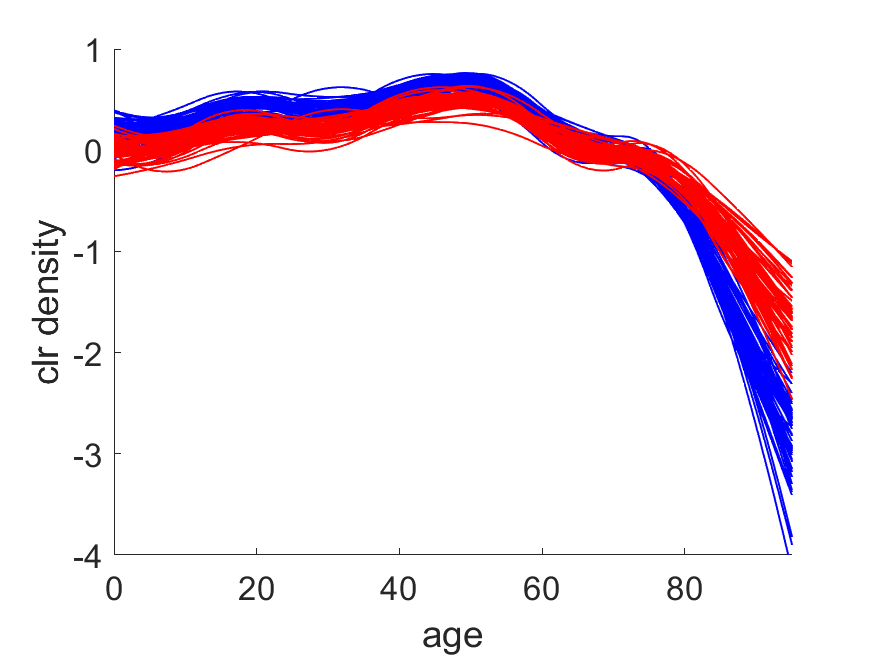}&
\includegraphics[width=0.5\textwidth]{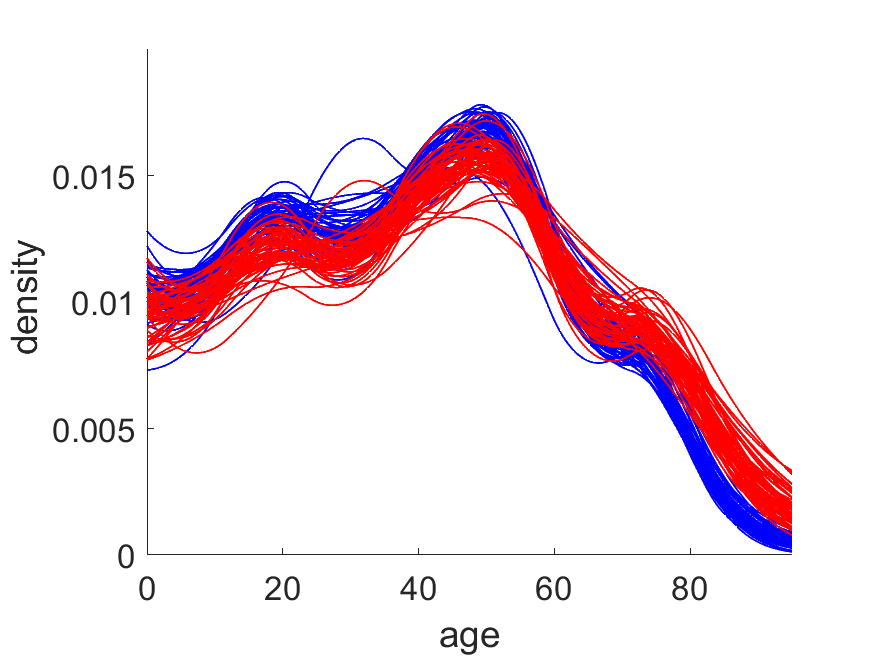}\\
\end{tabular}
\caption{Age distribution dataset - spline representation of the data in $L^2_{0}$ (left) and their $\mathcal{B}^{2}$ counterparts (right). 
}
\label{fig:example_JM}
\end{figure}

\subsubsection*{Spline representation of the main principal component in SFPCA}

To investigate the proposed approach for constructing the orthogonal basis within the real example setting, the impact of different factors was examined.
At first, the emphasis was given on the type of orthogonalization approach for $Z\!B$-spline basis and the effect of the individual basis functions on the principal functional components in SFPCA. Orthogonal spline basis obtained through the symmetric two-sided and both one-sided Gram-Schmidt orthogonalizations and $Z\!B$-splinets dyadic orthogonalization 
were all considered here. Furthermore, two equidistant sequences of knots of different lengths were chosen for further comparison:
\begin{align*}
&\Delta\lambda_{1} =\{\lambda^1_i\}_{i=0}^{8}, \quad \mbox{where} \quad \lambda_i^1 = 95i/8, \; i=0,\ldots,8,\\
&\Delta\lambda_{2} = \{\lambda^2_i\}_{i=0}^{20}, \quad  \mbox{where} \quad \lambda_i^2 = 95i/20, \; i=0,\ldots,20.
\end{align*}

\begin{figure}[ht!]
\centering
\begin{tabular}{cc}
\includegraphics[width=0.5\textwidth]{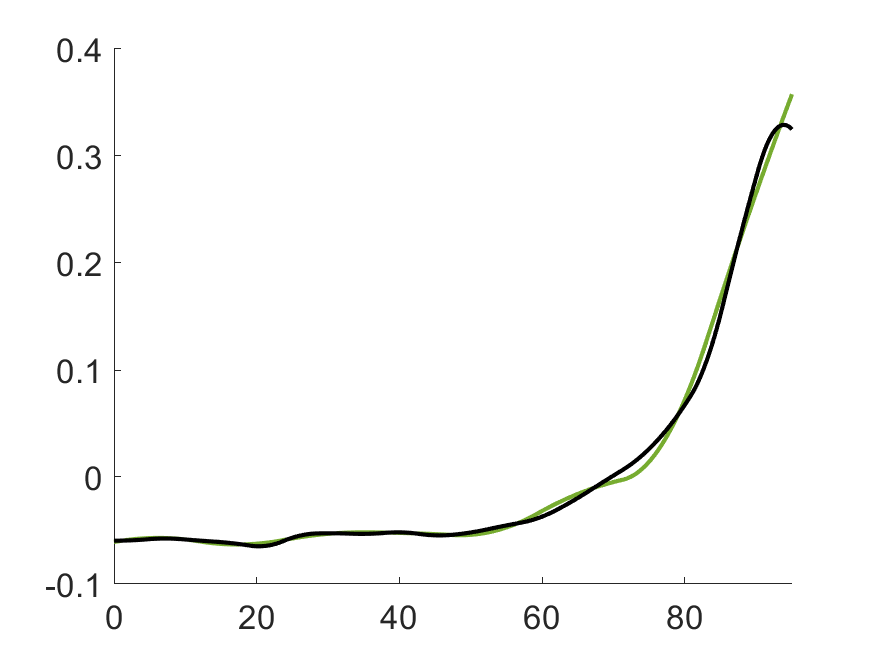}&
\includegraphics[width=0.5\textwidth]{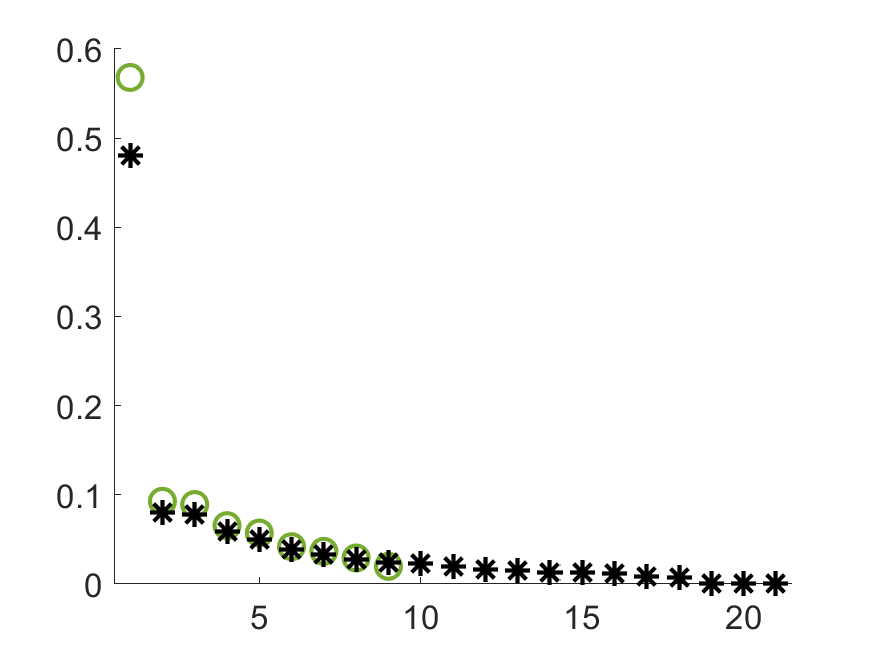}\\
\end{tabular}
\caption{The first principal component (left) and explained variability of eigenvalues (right) obtained through SFPCA for $\Delta\lambda_{1}$ (green) and $\Delta\lambda_{2}$ (black).}
\label{PCA_component1}
\end{figure}

The first principal component and the behaviour of all eigenvalues are shown in Figure \ref{PCA_component1} for both $\Delta \lambda_{1}$ and $\Delta \lambda_{2}$. The left figure confirms that for both choices of knots, the shape of the first functional principal component stays consistent. Note that the 
principal components do not depend on the
type of (orthogonal) basis, only on the chosen sequence of knots.
The screeplot of eigenvalues (Figure \ref{PCA_component1}, right) suggests the dominance of the first component as it is responsible for 56.8\% and 48.0\%  of the original variability contained in the data (where the explained variability is obtained as the ratio of the first eigenvalue to the sum of all eigenvalues) for $\Delta \lambda_1$ and $\Delta\lambda_2$, respectively.
Moreover, the shape of the first functional principal components also corresponds to the natural expectation about the main source of variability in the data, leading to a quite straightforward interpretation. Due to the densities being measured separately for men and women, one can expect the main mode of variability to be present in the higher age region while capturing both differences between men and women life expectancy and variability in life expectancy across regions.
This is confirmed by the larger deviation of the eigenfunction from zero in the age region from ca 75 to 95. 

Regardless of the type of orthogonal basis, we always get the same smoothing spline corresponding to the main principal component (eigenfunction), because it is determined uniquely. However, for the computations of its coefficients $\mathbf{o}^*$ from~(\ref{min}) we need to enumerate the matrix $\mathbf{N}_{kl}$ and the collocation matrix $\mathbf{O}_{k+1}(\mathbf{x})$. 
Here the effect of~the chosen basis is already visible in the numbers of non-zero elements that are given in Table \ref{nonzeroN}.  The computational efficiency of the $Z\!B$-splinet approach is noticeable for the second choice of knots $\Delta \lambda_{2}$ with nearly half of the non-zero elements of matrix $\mathbf{N}_{kl}$ and considerable lower number of non-zero elements of matrix $\mathbf{O}_{k+1}(\mathbf{x})$ compared to the one-sided GS approaches. $Z\!B$-splinet approach is comparable with two-sided GS approach for $\Delta \lambda_{1}$, but better in efficiency for $\Delta \lambda_{2}$. The sparsity patterns 
of the considered matrices  for the choice of knots $\Delta \lambda_{2}$ are plotted in Figure~\ref{fig:nonzero} for all four orthogonalization approaches.

\begin{table}[h!]
  \centering
  \begin{tabular}{c||c|c||c|c}
    & \multicolumn{2}{c||}{$\mathbf{N}_{kl}$} & \multicolumn{2}{c}{$\mathbf{O}_{k+1}(\mathbf{x})$} \\ 
    \hline
    \textbf{approach} & $\Delta \lambda_{1}$ & $\Delta \lambda_{2}$ & $\Delta \lambda_{1}$ & $\Delta \lambda_{2}$  \\ 
    \hline 
    \hline
    one-sided GS from left to right & 81 & 441 & 122 & 238\\ 
    \hline
    one-sided GS from right to left & 81 & 441 & 121 & 234\\ 
    \hline
    two-sided GS & 63 & 297 & 100 & 174\\ 
    \hline
    $Z B$-splinet & 63 & 243 & 114 & 170
  \end{tabular}
  \caption{Number of non-zero elements of matrices $\mathbf{N}_{kl}$ and $\mathbf{O}_{k+1}(\mathbf{x})$ for both choices of knots $\Delta \lambda_{1}$ and $\Delta \lambda_{2}$.}
  \label{nonzeroN}
\end{table}

\begin{figure}[ht!]
\centering
\begin{tabular}{cccc}
\includegraphics[width=0.22\textwidth]{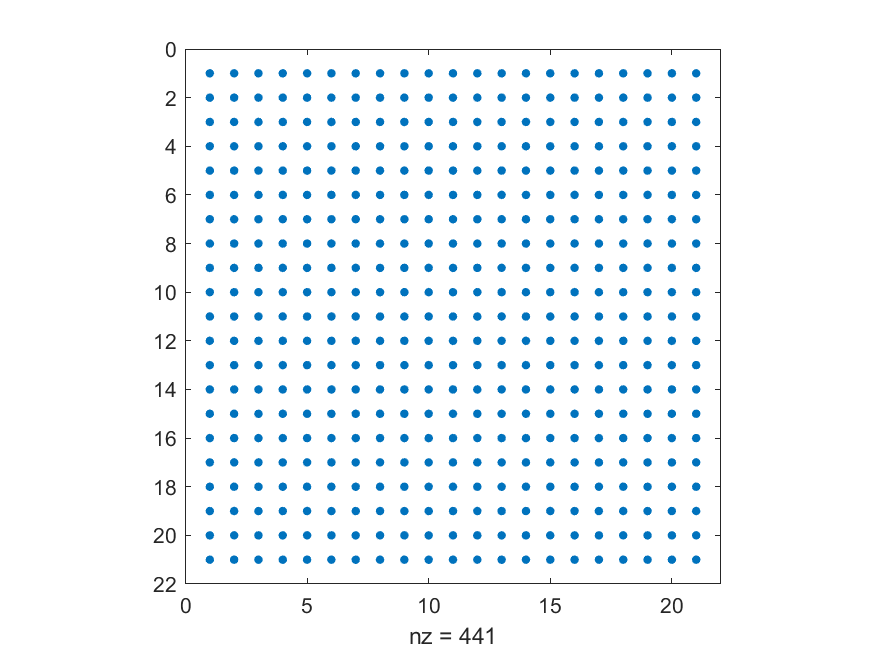} &
\includegraphics[width=0.22\textwidth]{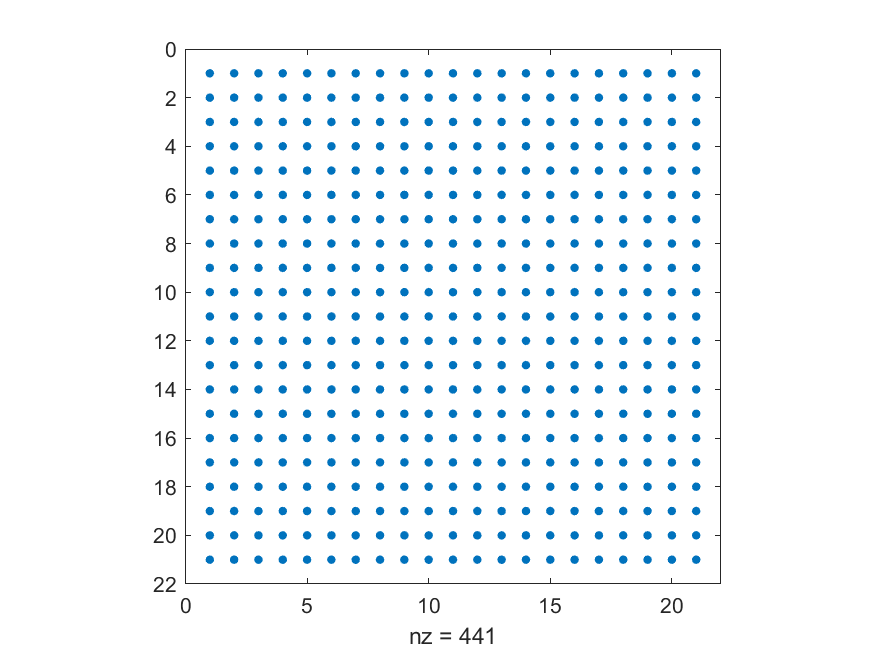} &
\includegraphics[width=0.22\textwidth]{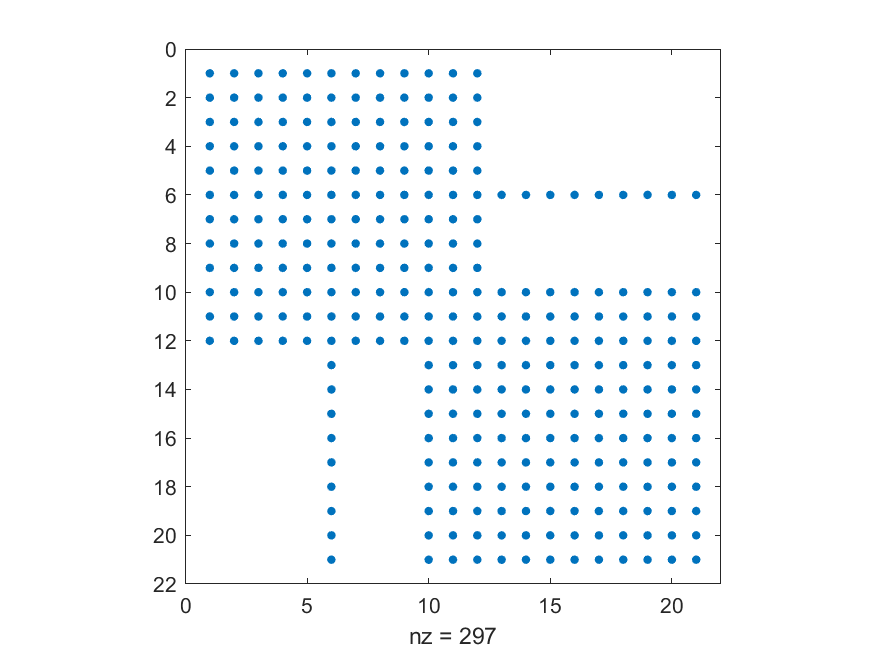} &
\includegraphics[width=0.22\textwidth]{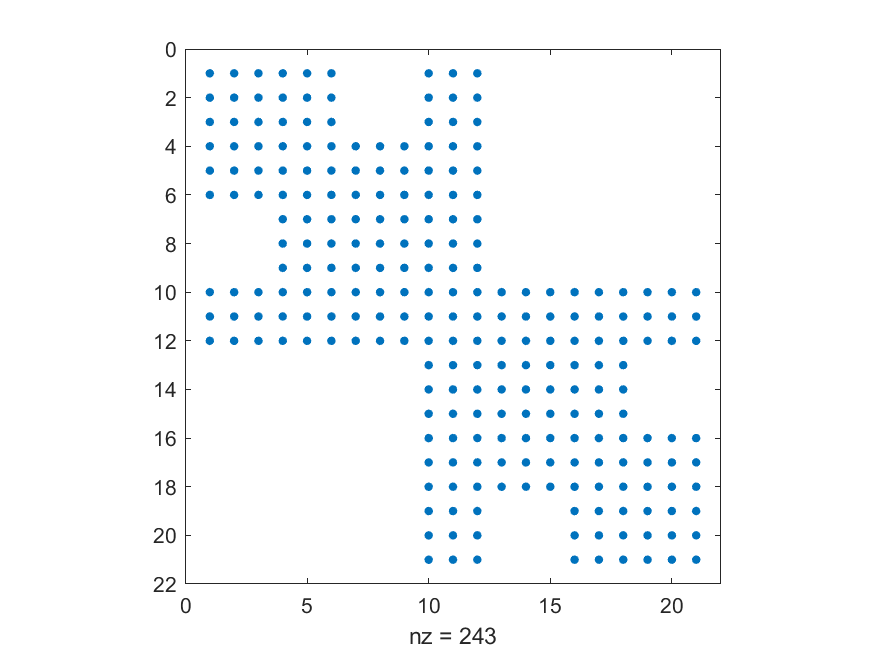}\\
\end{tabular}
\begin{tabular}{cccc}
\includegraphics[width=0.22\textwidth]{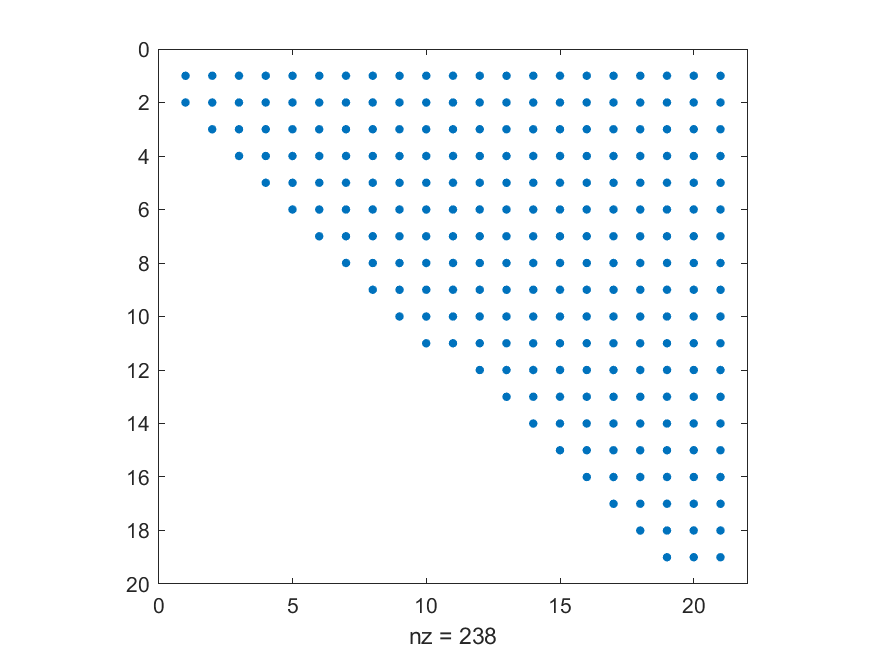} &
\includegraphics[width=0.22\textwidth]{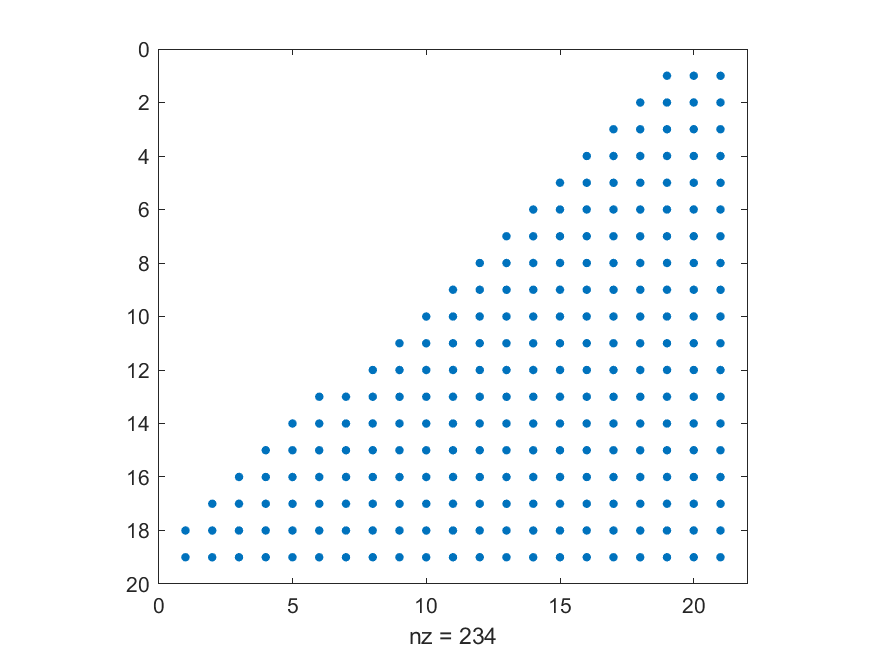} &
\includegraphics[width=0.22\textwidth]{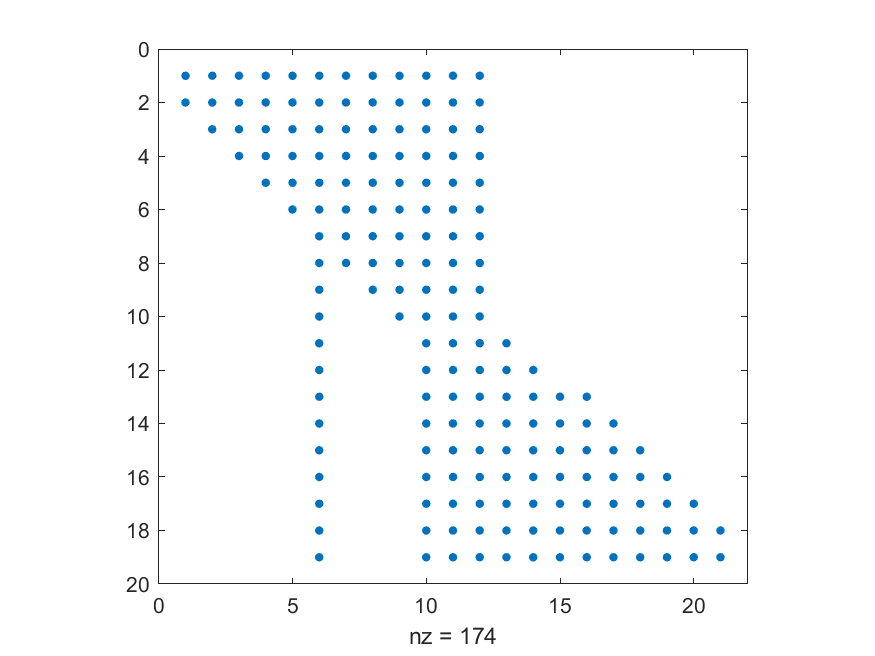} &
\includegraphics[width=0.22\textwidth]{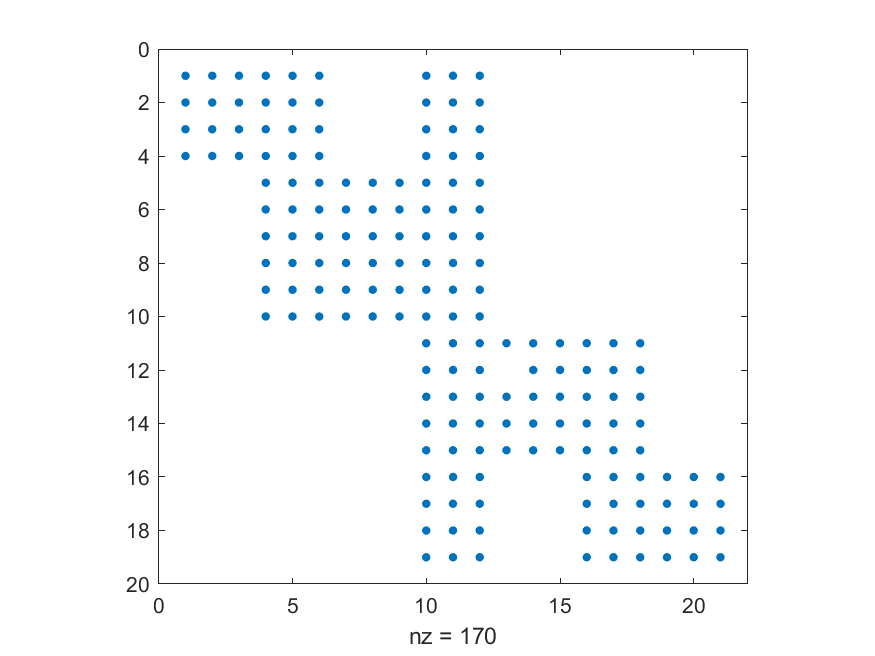}\\
\end{tabular}
\caption{Non-zero elements of matrices $\mathbf{N}_{kl}$ (first row) and $\mathbf{O}_{k+1}(\mathbf{x})$ (second row) for the choice of knots $\Delta \lambda_{2}$.
}
\label{fig:nonzero}
\end{figure}

\subsubsection*{Importance of basis functions in SFPCA for different 
orthogonalization approaches}
The focus of this subsection is on the contribution of individual basis functions constructed through different orthogonal approaches to the final form of the first principal component in SFPCA. 
The importance of basis functions corresponds to the absolute value of their respective coefficient, an element of the eigenvector of the sample covariance matrix of spline coefficients. 
Naturally, the higher the absolute value of the coefficient, the larger the effect of the corresponding basis function on the final estimate of the functional principal component. 

The idea of this exploratory procedure is to reduce the number of basis functions used in the expansion and remove those that do not contribute significantly to the first principal component. 
The rule to select a basis function for the approximation is set based on the absolute value of its coefficients to be greater than 0.1 (which is a default choice, e.g. when working with package \textit{stats} in statistical software R \cite{rcore21}). We mark the selected basis functions as active with respect to the first principal component.
The number of active basis functions indicates how well the spline basis is able to characterize the eigenfunction.  The lower number of active basis functions suggests that the basis is well-suited for the expansion and can capture the source of data variability by fewer elements.

The active basis functions for different orthogonal approaches are displayed in Figure \ref{PC1_significance_JM} for both choices of knots $\Delta \lambda_{1}$ (left) and $\Delta \lambda_{2}$ (right). 
The corresponding numbers of active basis functions for each approach are displayed in Table \ref{activ}. 
In all cases, some inactive basis functions were present as the total number of basis functions is 9 for $\Delta \lambda_{1}$ and 21 for $\Delta \lambda_{2}$. 
Since the goal is to decrease the number of active basis functions, 
the one-sided Gram-Schmidt orthogonalization from right to left is not convenient 
as the number 
of used active basis functions is the highest, leading to a more complex structure. The difference among the other orthogonalization approaches is more evident for the sequence of more knots where the spline basis is composed of more elements.  
For  $\Delta \lambda_{2}$, the number of active basis functions decreases from 21 to 7 for one-sided Gram-Schmidt approach from left to right and for $Z\!B$-splinets. 

\begin{table}[h!]
  \centering
  \begin{tabular}{c||c|c}
    \textbf{approach} & $\Delta \lambda_{1}$ &$\Delta \lambda_{2}$ \\ 
    \hline 
    \hline
    one-sided GS from left to right & 4 & 7 \\ 
    \hline
    one-sided GS from right to left & 8 & 18\\ 
    \hline
    two-sided GS & 5 & 10 \\ 
    \hline
    $Z\!B$-splinet & 5 & 7 
  \end{tabular}
  \caption{Summary of the number of active basis functions for the first principal component for $\Delta \lambda_{1}$ and $\Delta \lambda_{2}$.}
  \label{activ}
\end{table}

\begin{figure}[ht!]
\centering
\begin{tabular}{cc}
\includegraphics[width=0.4\textwidth]{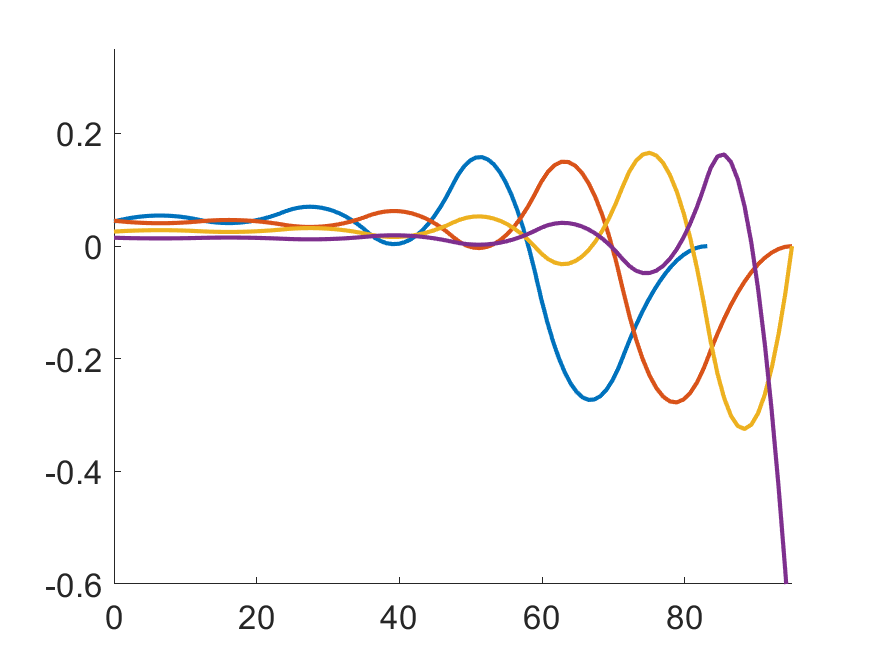}&
\includegraphics[width=0.4\textwidth]{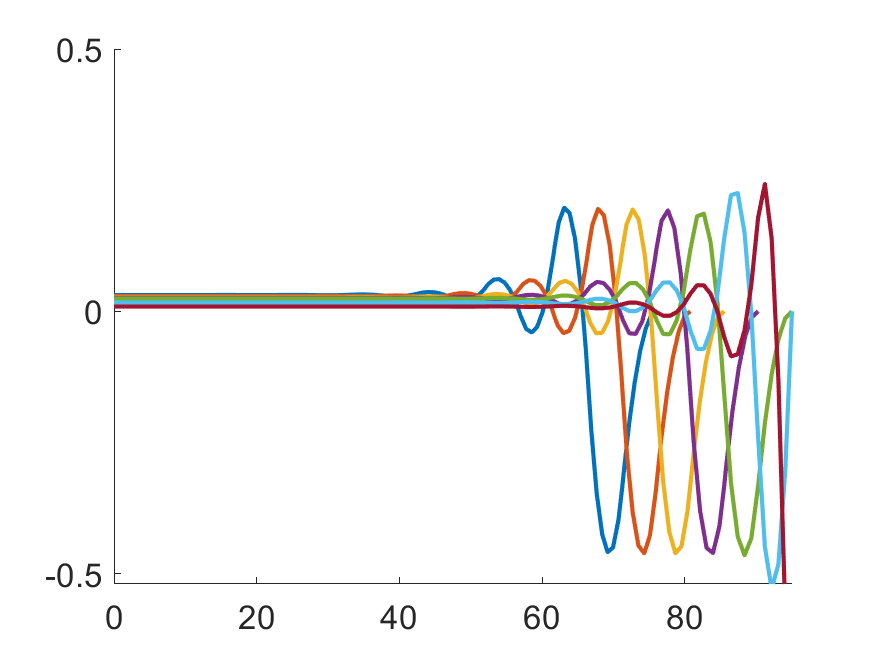}\\
\includegraphics[width=0.4\textwidth]{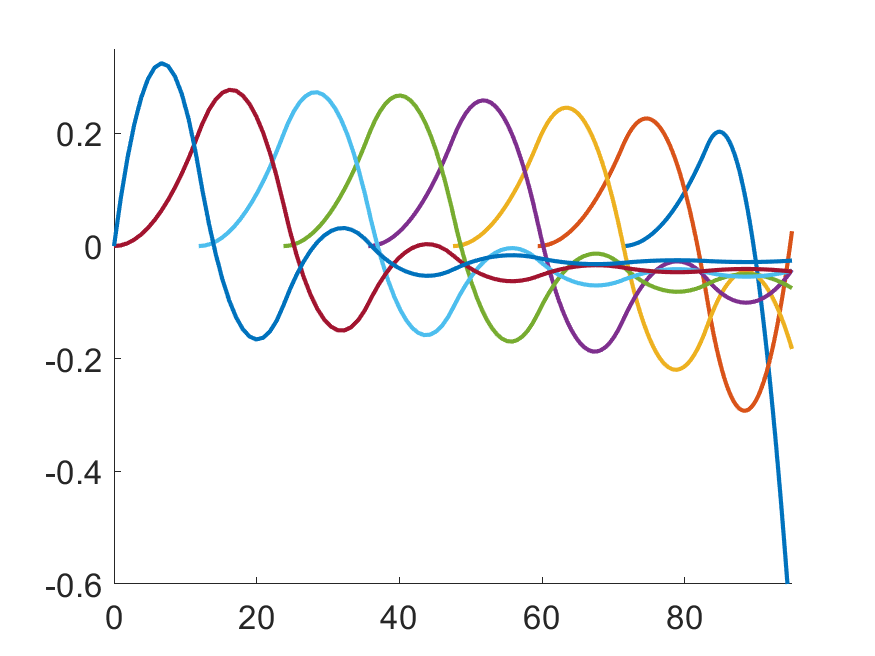}&
\includegraphics[width=0.4\textwidth]{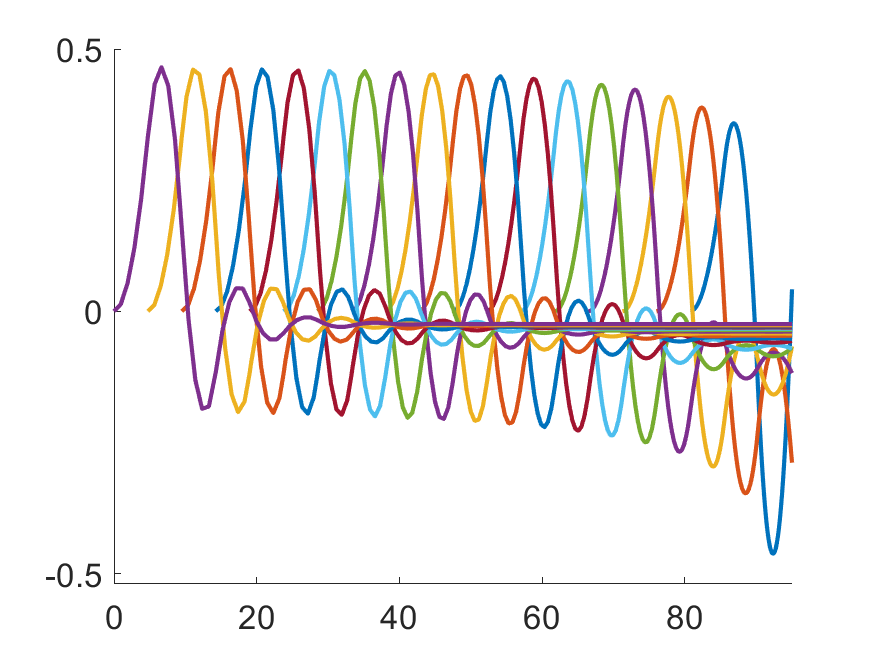}\\
\includegraphics[width=0.4\textwidth]{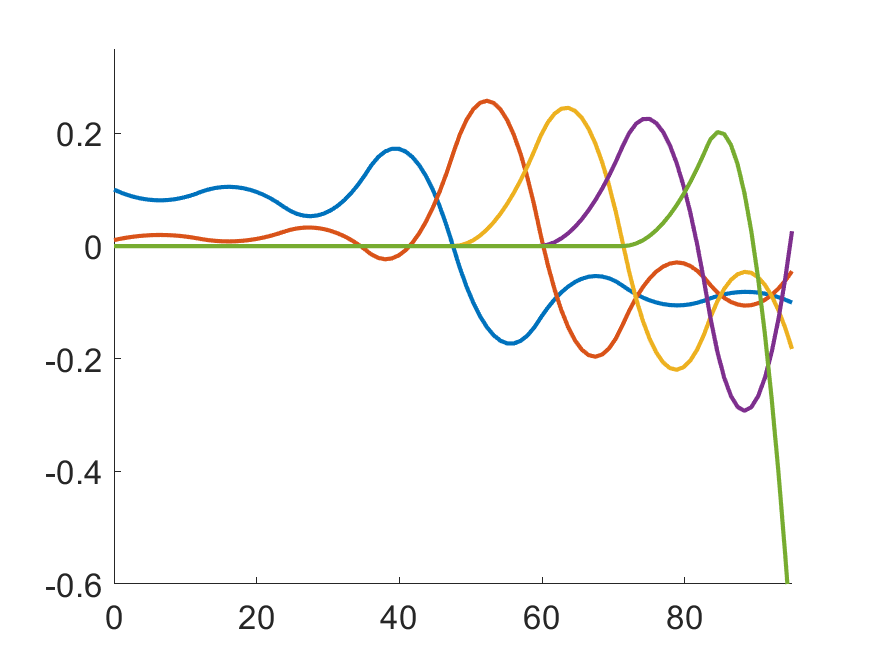}&
\includegraphics[width=0.4\textwidth]{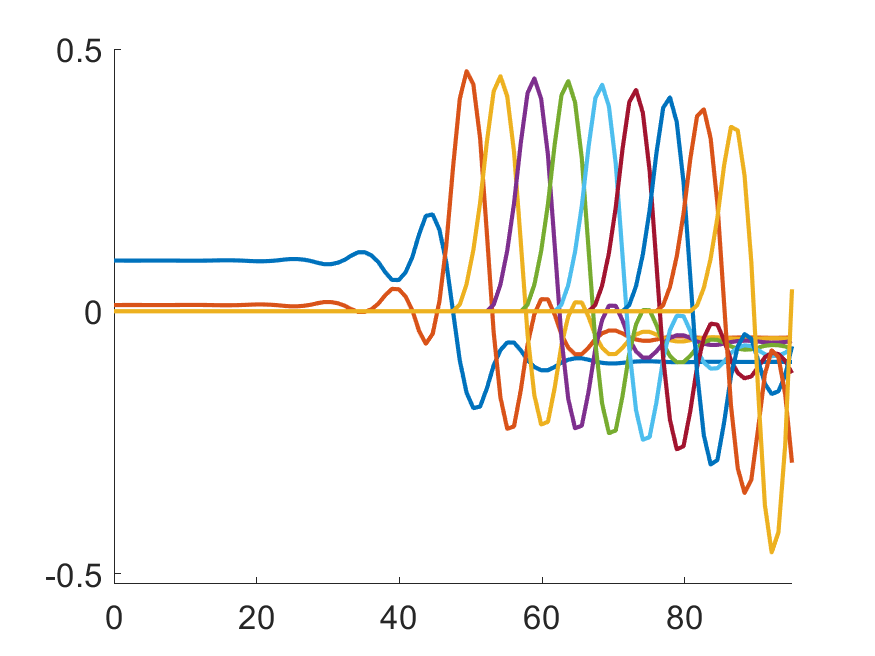}\\
\includegraphics[width=0.4\textwidth]{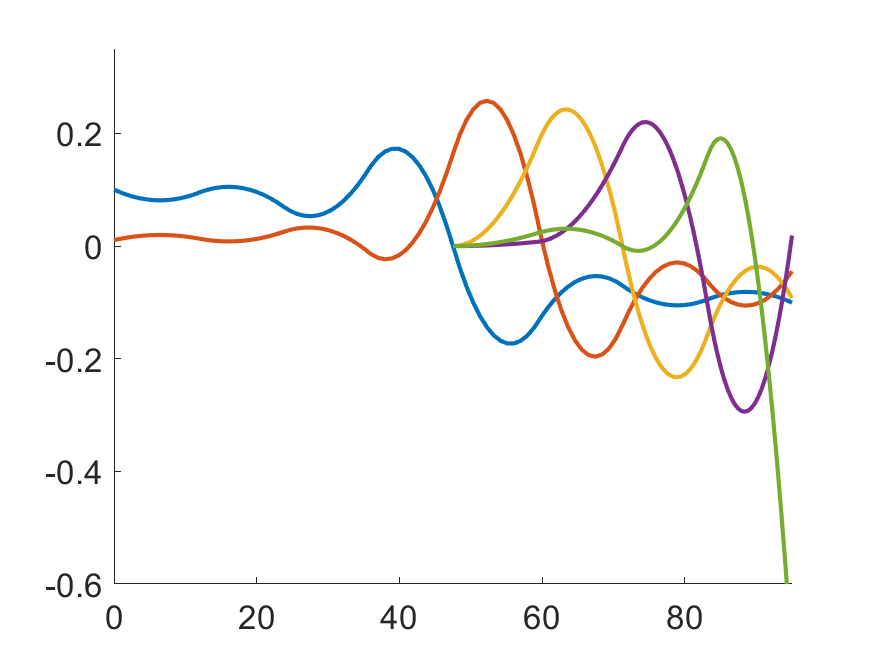}&
\includegraphics[width=0.4\textwidth]{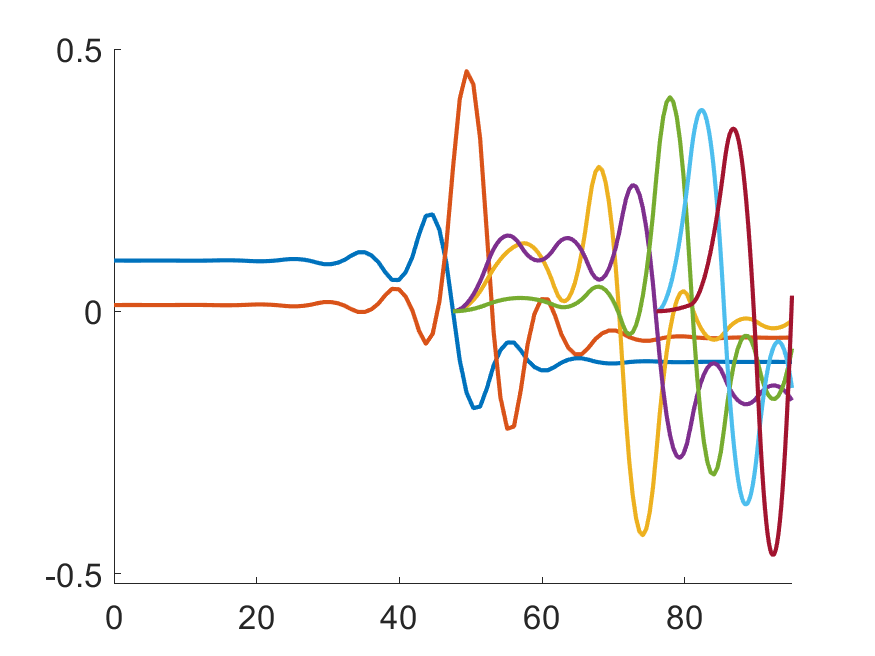}\\
\end{tabular}
\caption{The functional bases used for SFPCA with $\Delta \lambda_{1}$ (left) and $\Delta \lambda_{2}$ (right).  The first row corresponds to  one-sided GS from left to right,  the second to one-sided GS from right to left,  the third to two-sided GS and the last one to $Z\!B$-splinets.}
\label{PC1_significance_JM}
\end{figure}

In order to explain the dominance of the one-sides Gram-Schmidt approach from left to right over its opposite counterpart from right to left, we can argue as follows. The observed concentration of variability in the higher age group implies a need for more significant basis functions carrying this local information, i.e. functions having large variations in the higher age group interval while having small variations in the lower age group interval. Since this is exactly the case for the one-sided Gram-Schmidt approach from left to right, this orthogonal basis is well-suited for the considered dataset. The opposite would be true when the variability in the data would be concentrated in the lower age group. Therefore, a versatile orthogonalization approach for a general case, such as two-sided Gram-Schmidt or $Z\!B$-splinet approach, is required. 
When comparing the two-sided Gram-Schmidt orthogonalization approach and $Z\!B$-splinets, there is not a big difference in the number of active basis functions. However, the $Z\!B$-splinet approach gains in efficiency due to the local support of the corresponding basis functions. Moreover, the active basis functions for $Z\!B$-splinet approach follow the variability in the higher age group interval. Only those are selected that have support in this interval, as observed from the dyadic structure in Figure \ref{pyramids}.

\begin{figure}[ht!]
\centering
\begin{tabular}{cc}
\includegraphics[width=6.5cm, height=9cm]
{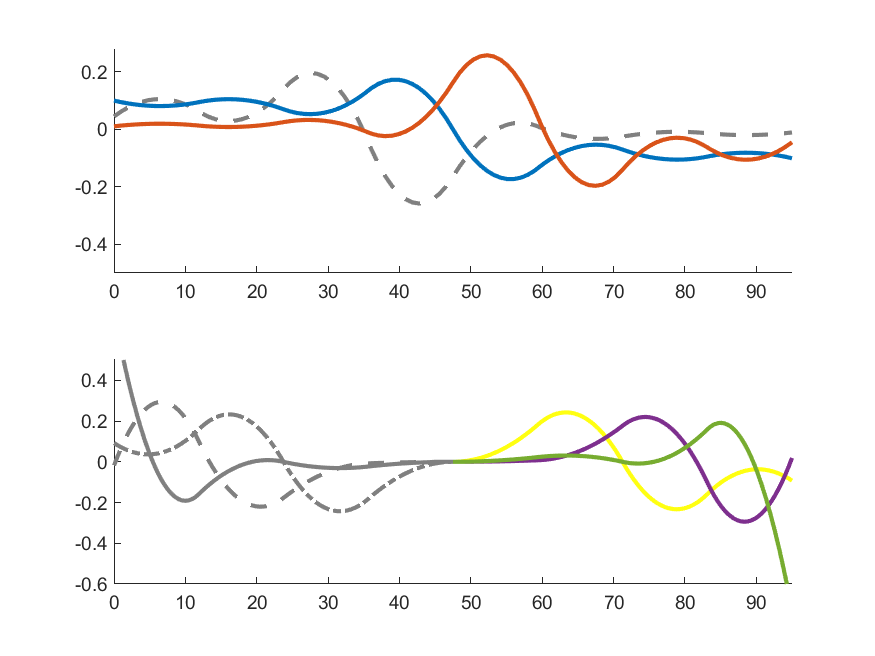} &
\includegraphics[width=6.5cm, height=9cm]
{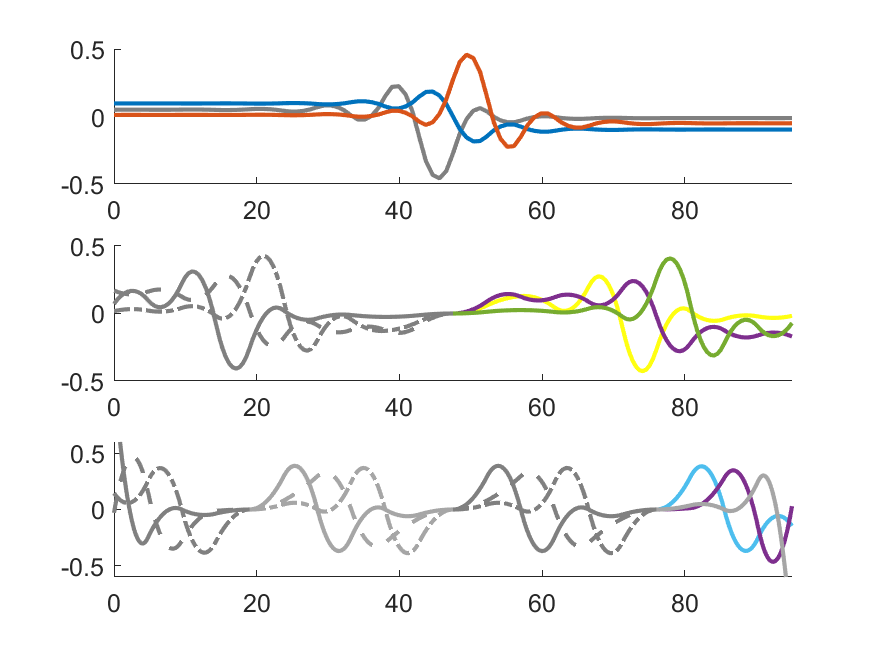}
\end{tabular}
\caption{Dyadic structure of active basis functions for $Z\!B$-splinets constructed on $\Delta \lambda_{1}$ (left) and $\Delta \lambda_{2}$ (right).
}
\label{pyramids}
\end{figure}

For comparison of the resulting active basis expansions and also to admit possible limitations of our approach, we provide the corresponding approximations of the first principal component, see Figure \ref{appr_FPC}. Clearly, the approximation power of the reduced expansion declines with a decreasing number of active basis functions. Therefore, we cannot expect to profit from the significant reduction and improve the accuracy of the approximation simultaneously. However, the relative errors measured in $L_2$ norm remain within a reasonable range for all orthogonalization approaches, see Table \ref{errortab}. Although the relative error for the $Z\!B$-splinet approach is the highest observed, we must note that it mainly corresponds to the low number of active basis functions; hereby, the one-sided GS from left to right clearly benefits from specific data structure, which would however hardly be so in the general case. The latter one-sided GS approach from right to left suffers from overparametrization (cf. Table \ref{activ}). Thus, the real competitor here is two-sided GS, which better approximates the first principal component at the cost of a higher number of active basis functions.

\begin{figure}[ht!]
\centering
\begin{tabular}{cc}
\includegraphics[width=0.4\textwidth]{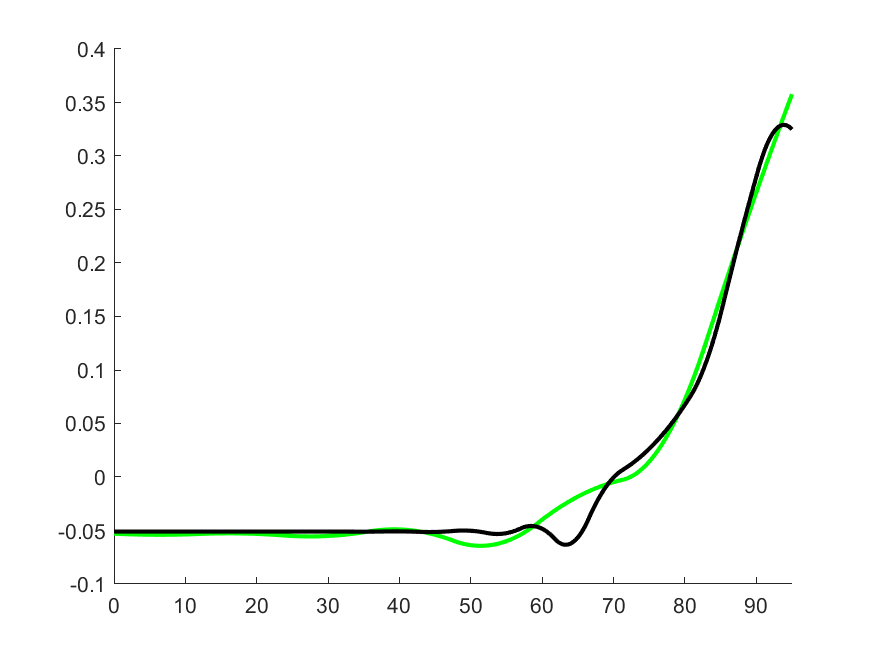}&
\includegraphics[width=0.4\textwidth]{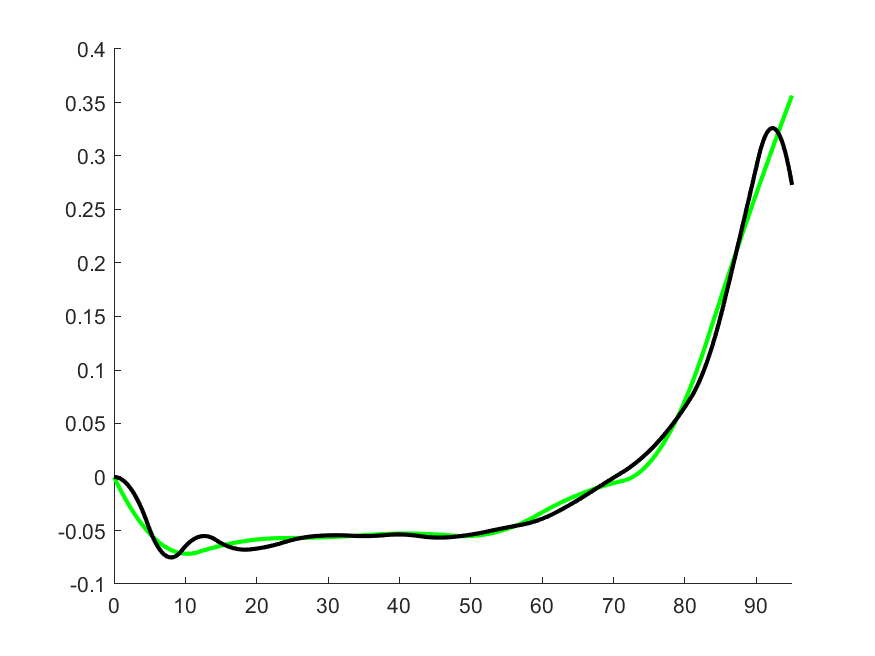}\\
\includegraphics[width=0.4\textwidth]{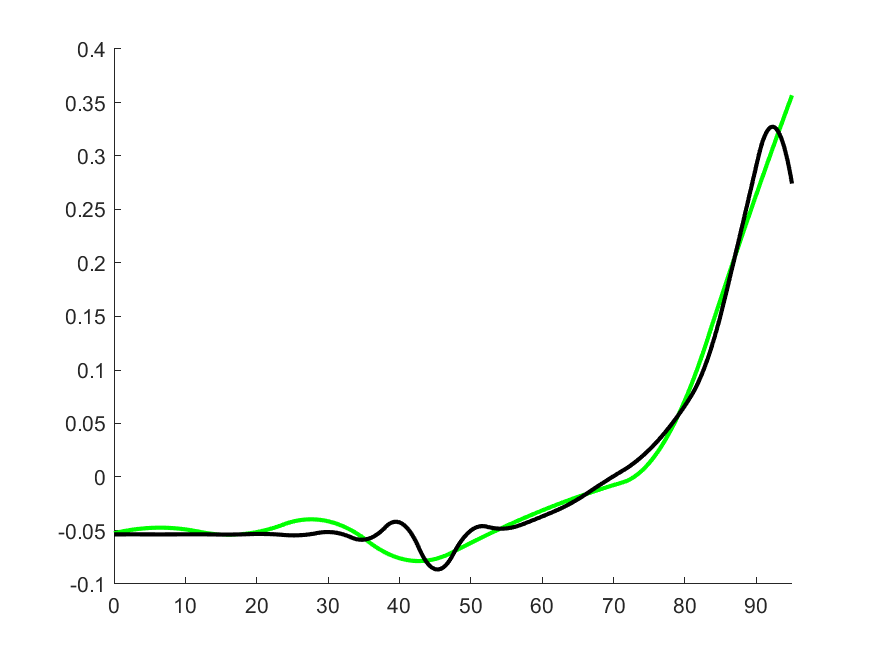}&
\includegraphics[width=0.4\textwidth]{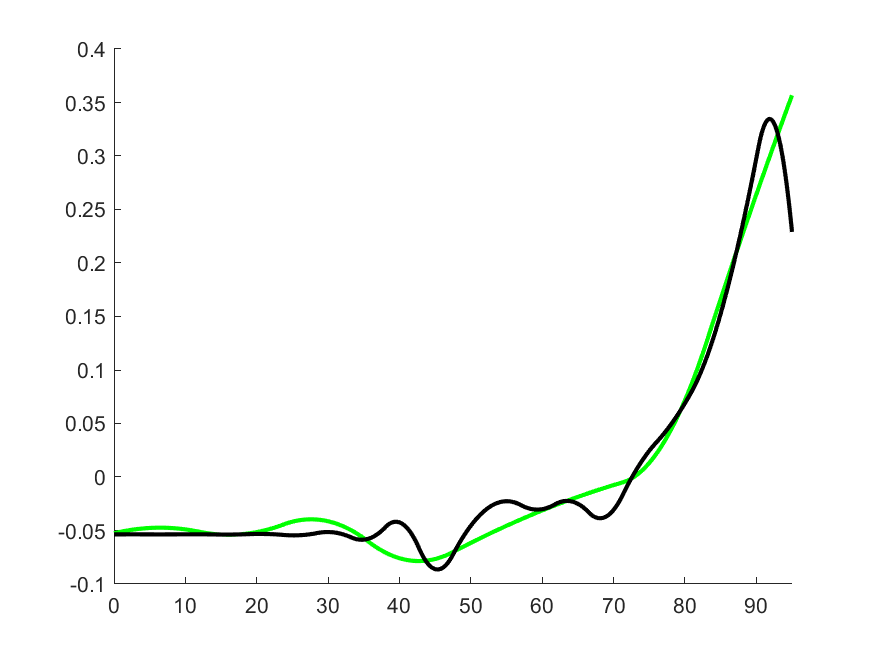}
\end{tabular}
\caption{Approximation of the main principal component by one-sided GS from left to right (top left), one-sided GS from right to left (top right), two-sided GS (bottom left) and \mbox{$Z\!B$-splinets} (bottom right)  for $\Delta\lambda_{1}$ (green) and $\Delta\lambda_{2}$ (black).}
\label{appr_FPC}
\end{figure}

Rather than analysing the  $L_2$-error and its sources, based on the default criterion of active basis functions, we prefer to introduce a more sophisticated tool to reduce the basis expansion and to compare the suitability of each orthogonalization approach in terms of explained variability in the next section.

\begin{table}[h!]
  \centering
  \begin{tabular}{c||c|c}
    \textbf{approach} & $\Delta \lambda_{1}$ &$\Delta \lambda_{2}$ \\ 
    \hline 
    \hline
    one-sided GS from left to right & 0.0542 & 0.0982 \\ 
    \hline
    one-sided GS from right to left & 0.0818 & 0.1158\\ 
    \hline
    two-sided GS & 0.1001 & 0.0854 \\ 
    \hline
    $Z\!B$-splinets & 0.1001 & 0.1424 
  \end{tabular}
  \caption{Relative errors of approximated first principal component (Fig. \ref{appr_FPC})
  }
  \label{errortab}
\end{table}

\subsubsection*{Effect of sparsity on explained variability}
Comparison of functions, resulting from
aggregation of selected (active) basis functions, gives us a first insight on quality of approximation of the first functional principal component. 
However, their effect can be explored in a different sense by using sparse principal component analysis (SPCA) \cite{erichson18}, applied in the case of an orthogonal spline basis simply to spline coefficients. Accordingly, 
the principal components are estimated using a lower number of 
basis functions by suppressing the redundant ones. 
These components then can be, in the above sense, considered as sparse (or reduced) principal components (sparse FPCA is actually defined a bit differently \cite{NIE2020}). Generally, the higher the sparsity parameter in SPCA, 
the lower number of active basis functions is used. The price to pay for the simpler structure of the principal component is the natural decline of the explained variability of the given component. Here, SPCA is used as a more sophisticated tool for showing the significance of the effect of the basis functions on the principal components, resp. how well the spline basis is able to capture the specific functional data structure.  
To see the effect of using a sparse approach, 
the results were investigated for different values of sparsity parameter on the interval $\left[ 0,1\right]$. 
The evolvement of the number of active basis functions is depicted in Figure \ref{nonzeros}. From this, one can see a~relatively fast decline in the number of active basis functions for all approaches, meaning that the full structure of the first functional principal component (FPC), i.e. eigenfunction,  is soon replaced by its estimate using a significantly lower number of basis functions. While for $\Delta \lambda_{1}$, the one-sided Gram-Schmidt approach from left to right holds up the longest with at least one active basis function (i.e. before the method breaks down) compared to the other approaches, for $\Delta \lambda_{2}$ both $Z\!B$-splinets and one-sided Gram-Schmidt from left to right offer comparable results as the most stable approaches in the sense of resistance to sparsity. 

\begin{figure}[ht!]
\centering
\includegraphics[width=1.0\textwidth,trim={0cm 1.5cm 0cm 0.5cm},clip]%
{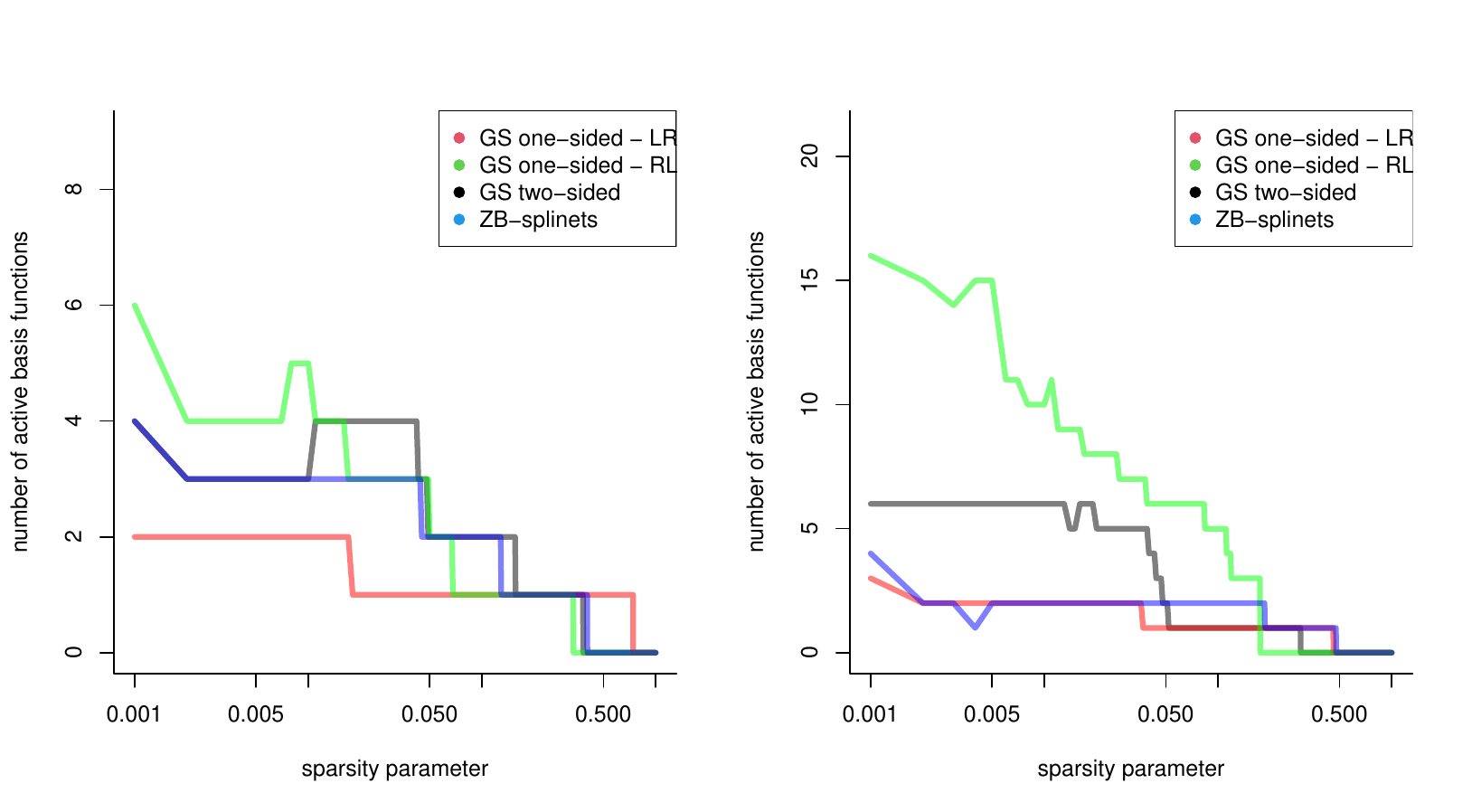}
\caption{The number of active basis functions forming the first principal component for $\Delta \lambda_{1}$ and $\Delta \lambda_{2}$ with respect to different values of sparsity parameter depicted for different orthogonalization approaches.
}
\label{nonzeros}
\end{figure}

Figure \ref{FVE} offers further insight and comparison of the methods as it shows the proportion of the overall variability, explained by the first sparse principal component, depending on the choice of the sparsity parameter. Here, the sparsity parameter is restricted to the values where all approaches lead to a nonzero first principal component.
It is interesting here to see that, while for the first FPC one-sided Gram-Schmidt orthogonalization from left to right dominates significantly with its maintained variability for $\Delta\lambda_{1}$, both two-sided Gram-Schmidt and especially $Z\!B$-splinets offer comparable results for $\Delta\lambda_{2}$. 
We note that the dominance of one-sided GS from left to right, especially for the choice of $\Delta\lambda_{1}$, is not surprising since the nature of original data corresponds to the nature of this orthogonalization approach as observed also earlier. However, the $Z\!B$-splinet approach, that gives similar results,  is a more versatile tool for general data.

\begin{figure}[ht!]
\centering
\includegraphics[width=1.0\textwidth,trim={0cm 1.5cm 0cm 0.5cm},clip]
{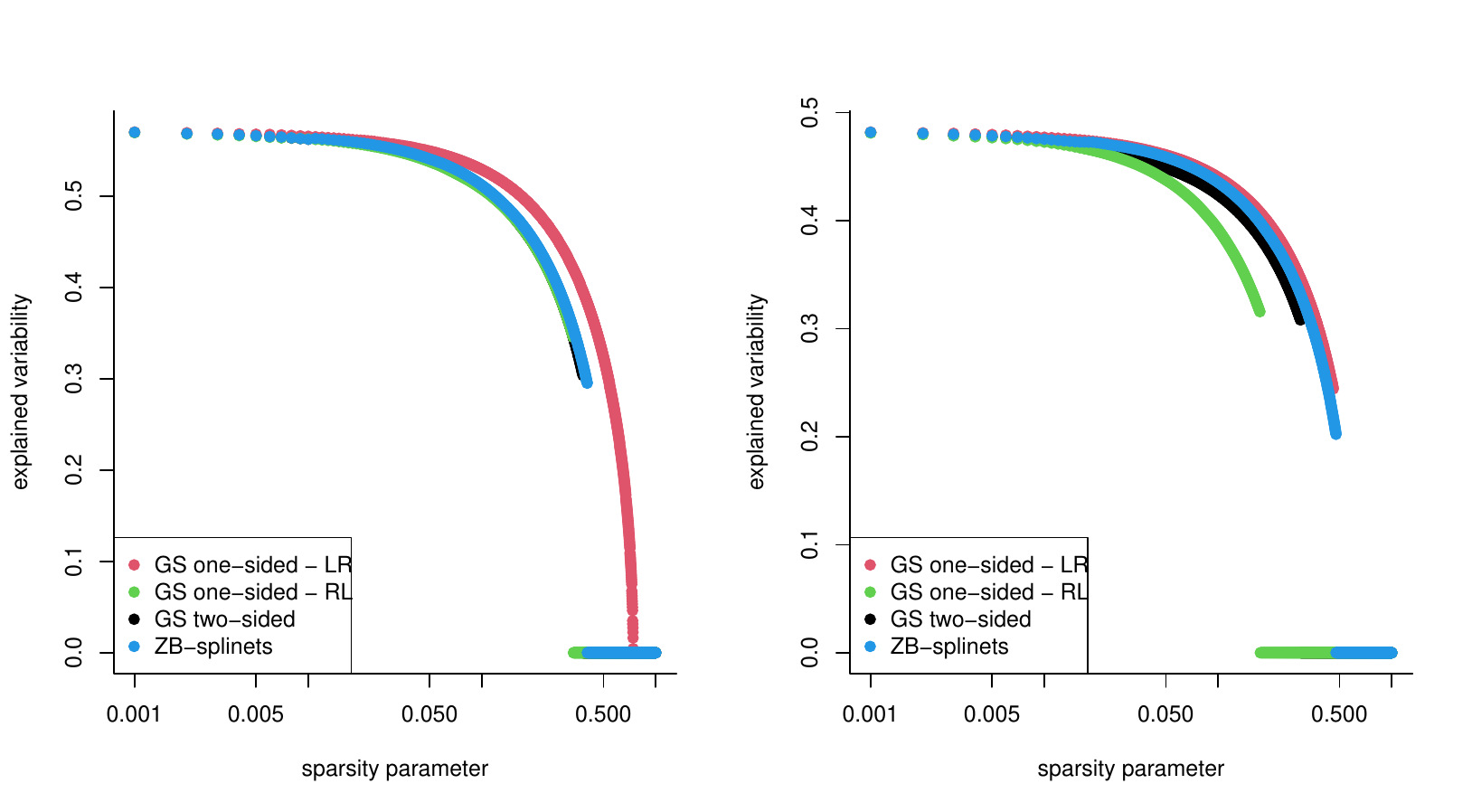}
\caption{Fraction of explained variability by the first principal component for different orthogonalization approaches.}
\label{FVE}
\end{figure}

\subsection{Results and discussion}
 In this application part, four approaches for construction of the orthogonal spline basis were compared from several points of view. 
 We conclude that both the $Z\!B$-splinets and one-sided Gram-Schmidt from left to right have individual basis functions that better correspond to the first FPC than the other approaches. This means that a lower number of (active) basis functions is needed for a relatively accurate description of the component. This phenomenon is even more prominent in the case with higher number of knots ($\Delta \lambda_{2}$). While the left-to-right Gram-Schmidt clearly takes advantage from the shape of the first component, i.e. varied behavior on the right side of the overall support, $Z\!B$-splinets utilize the basis functions with a local support within the same area. As there are more locally defined basis functions with increasing number of knots, it can be expected that the latter approach should prevail even more in such setting. Furthermore, the existence of locally defined basis functions on the whole domain for functional data ensures that the efficiency of $Z\!B$-splinets does not depend on the shape of the (first) principal component.
The dominance of the left-to-right Gram-Schmidt and $Z\!B$-splinets in this example is further emphasized 
as with the increasing value of the sparsity parameter both approaches endure the longest without breaking down and therefore keeping a nonzero amount of explained variability within the first FPC. 
While the other aspects put the two discussed approaches side by side, the indisputable advantage of the $Z\!B$-splinets is in its computational efficiency.  
Although, clearly, we present just one concrete empiric example here, this was carefully chosen to clearly demonstrate both advances and possible limitations of the $Z\!B$-splinets approach.

\section{Conclusions}
We have introduced a~new orthogonal spline basis for the functional representation of probability density functions. Due to their specific properties of scale invariance and relative scale, probability density functions need a suitable functional representation that respects their relative nature. This is achieved with the Bayes spaces methodology which also enables to convert processing of PDFs to the standard space of square-integrable Lebesgue functions by the centered log-ratio transformation. Since this transformation induces a zero integral constraint, a~new $Z\!B$-spline basis reflecting this condition was recently developed for a suitable basis expansion of probability density functions. In practical applications such as PCA, the orthogonality of the basis plays a crucial role in the proper function representation and interpretability of the results. 

The construction of an orthogonal basis for the representation of probability density functions was addressed in this paper and an effective approach for orthogonalization of $Z\!B$-splines was introduced. The proposed $Z\!B$-splinet approach has shown to maintain several advantages. Above all, $Z\!B$-splinets were proved to be beneficial in terms of
\begin{enumerate}
    \item computational efficiency,
    \item locality of spline supports.
\end{enumerate}
In particular, the computational efficiency was measured by the number of evaluations of inner products needed for the basis orthogonalization. The locality was defined as small relative total support. 
This enables  $Z\!B$-splinets to have the potential for flexible adaptation for functions with local characteristics.
These aspects of $Z\!B$-splinets were compared to one-sided Gram-Schmidt (from left to right and from right to left) and with a two-sided Gram-Schmidt orthogonalization approach. $Z\!B$-splinet approach has shown to surpass the others in both criteria: having the lowest number of inner product evaluations and the smallest relative total support.
The proposed approach was demonstrated on an empirical demographic example using functional PCA as one of statistical tools where orthogonality of the basis is essential. Focusing on the first principal component, the presented results show that $Z\!B$-splinets offer comparable results concerning the preservation of data variability while maintaining both advantages of flexibility and computational efficiency.
To enhance the use of $Z\!B$-splinets among broader audience, the plan is to prepare a comprehensive R package devoted to the construction of $Z\!B$-splinets in the context of functional data analysis.

\label{conclusions}

\section*{Acknowledgements}
{\bf Funding:}\\
We gratefully acknowledge the support of this research and researchers by the following grants: IGA PrF 2024 006 Mathematical models and the Czech Science Foundation grant 22-15684L Generalized relative data and robustness in Bayes spaces.


\bibliography{references}

\end{document}